\documentclass[journal,10pt]{IEEEtran}

\usepackage{algorithm, algorithmic, amsthm, amsmath, amsfonts, datetime, math tools, graphicx, amsmath, amssymb, bm, amsbsy, breqn, accents, float, listings, courier, lscape, multirow, multicol, longtable, dsfont, bbold, cite, soul, color, tikz, scalerel,empheq, paralist, cleveref, enumitem, subcaption}
\usetikzlibrary{shapes, arrows}
\usetikzlibrary{patterns}
\usetikzlibrary{arrows.meta}
\usetikzlibrary{arrows,decorations.markings}
\usetikzlibrary{decorations.pathmorphing}
\newcommand{\R}{\mathbb{R}}

\newcommand{\scrH}{\mathcal{H}}

\newcommand{\scrL}{\mathcal{L}}

\newcommand{\scrU}{\mathcal{U}} 
\newcommand{\scrX}{\mathcal{X}}
\newcommand{\scrY}{\mathcal{Y}}
\newcommand{\scrZ}{\mathcal{Z}}

\newcommand{\p}{\intercal}
\newcommand{\compPowerRV}{X}

\newcommand{\compPowerSet}{\scrX}
\newcommand{\compPowerSetIndices}{\{1,2,\ldots,|\compPowerSet|\}}
\newcommand{\compPowerSetSize}{|\compPowerSet|}
\newcommand{\compPowerDistribution}{\pi}
\newcommand{\compPowerDistributionSpace}{\Pi}

\newcommand{\obsv}{y}
\newcommand{\obsvRV}{Y}
\newcommand{\obsvSet}{\scrY}
\newcommand{\obsvSetSize}{|\obsvSet|}
\newcommand{\info}{Z}
\newcommand{\infoSet}{\scrZ}
\newcommand{\reward}{r}
\newcommand{\transMatrix}{P}
\newcommand{\obsvMatrix}{B}
\newcommand{\action}{u}
\newcommand{\actionSet}{\{1,2\}}
\newcommand{\stopTime}{\tau}
\newcommand{\stopNumber}{l}
\newcommand{\numStops}{L}
\newcommand{\stopNumberSet}{\{1,2,\ldots,\numStops\}}
\newcommand{\policy}{\mu}
\newcommand{\totalReward}{J}
\newcommand{\discountFactor}{\rho}
\newcommand\Tau{\Gamma}

\newcommand{\numStopsCost}{\Omega}
\newcommand{\contset}{D}
\newcommand{\stopset}{M}
\newcommand{\bracketRound}[1]{\left(#1\right)}
\newcommand{\bracketSquare}[1]{\left[#1\right]}

\newcommand{\Prob}{\mathbb{P}}

\newtheorem{theorem}{Theorem}

\newtheorem{lemma}{Lemma}
\newtheorem{definition}{Definition}

\title{\LARGE \bf
Multiple-stopping time Sequential Detection for Energy Efficient Mining in Blockchain-Enabled IoT
}

\author{Anurag~Gupta\thanks{Anurag Gupta and Vikram Krishnamurthy are with the School of Electrical \& Computer Engineering, Cornell University, Ithaca NY, 14853, USA.  (e-mail: ag2589@cornell.edu; vikramk@cornell.edu).},
	and Vikram~Krishnamurthy,~\IEEEmembership{Fellow,~IEEE}}

\begin{document}

\maketitle
\thispagestyle{empty}
\pagestyle{empty}

\begin{abstract}
What are the optimal times for an Internet of Things (IoT) device to act as a blockchain miner? The aim is to
minimize the energy consumed by low-power IoT devices
that log their data into a secure (tamper-proof) distributed
ledger. We formulate a multiple stopping time Bayesian sequential detection problem 
to address energy-efficient blockchain mining for IoT devices. The objective is to identify $\numStops$ optimal stops for mining, thereby maximizing the probability of successfully adding a block to the blockchain;
we also present a model to optimize the number of stops (mining instants). The formulation is equivalent to a multiple stopping time POMDP. Since POMDPs are in general computationally intractable to solve, we show mathematically using submodularity arguments that the optimal mining policy has a useful structure: 1) it is monotone in belief space, and 2) it exhibits a threshold structure, which divides the belief space into two connected sets. Exploiting the structural results, we formulate a computationally-efficient linear mining policy for the blockchain-enabled IoT device. We present a policy gradient technique to optimize the parameters of the linear mining policy. Finally, we use synthetic and real Bitcoin datasets to study the performance of our proposed mining policy. We demonstrate the energy efficiency achieved by the optimal linear mining policy in contrast to other heuristic strategies.
\end{abstract}

\section{Introduction}
    Blockchain is a decentralized distributed ledger technology~\cite{2020:AS}. Each block in the chain contains a set of transactions and a cryptographic hash of the previous block. This creates a chain of blocks that are secure: it is difficult for a single entity to take control of the network or to alter past transactions. An important element of blockchain is Proof of Work (PoW)~\cite{2008:SN}. PoW is a consensus algorithm used in blockchain to add new blocks to the chain; this requires miners to solve a cryptographic puzzle. The first miner to solve the puzzle is rewarded monetarily.     
    However, PoW in blockchain requires a large amount of computational power and leads to high energy consumption. This is detrimental to using blockchain in an IoT application.
    
    We focus on blockchain-enabled IoT applications wherein IoT devices are resource-constrained \cite{2021:SM-et-al}, e.g. wireless sensor networks. The combination of blockchain and IoT can create a secure and decentralized network of devices, enabling efficient and transparent data sharing and transactions~\cite{2018:ON}. For example, in a sensor network, blockchain can provide an immutable and tamper-proof\footnote{Sensor networks are used in the field of agriculture to log the farming practices like the use of pesticides. An immutable and tamper-proof database will improve trust between farmers, manufacturers, retailers and consumers~\cite{2023:JB}.} data storage platform for storing sensor readings; the decentralized storage of data also makes it immune to a single point of failure~\cite{2023:EK}. An IoT network consists of heterogeneous devices: some have low power and low energy requirements, like Raspberry Pi, while others have high processing power and energy requirements, like PCs and servers. Low-power IoT devices are typically deployed for data collection, whereas high-power devices are used for time-critical applications. Irrespective of their processing power, the devices are capable of mining in a blockchain.
    For a typical IoT application~\cite{2021:SM-et-al}, CPU usage without mining is around 3\%-6\%. However, mining in blockchain increases the CPU usage of IoT devices to 30\%-50\%; this is roughly a ten times increase in power consumption. This is detrimental for IoT devices as they have limited energy resources. Additionally, blockchain mining has an adverse impact on the environment due to its energy consumption~\cite{2021:LB-et-al}. This motivates the study of energy-efficient mining strategy in blockchain for an IoT device: the IoT device wants to optimize its mining time instants so as to maximize its probability of adding a new block in the blockchain. By doing so, the IoT device prevents the waste of energy on mining when there is high competition to add a new block to the blockchain.
    
   \subsection*{Main Idea. Multiple Stopping Time POMDP} 
   The problem we address is: \textit{What are the optimal times for an IoT device to act as a blockchain miner?} The aim is to minimize the energy consumed by low-power IoT devices that register their data into a secure (tamper-proof) distributed ledger. In IoT applications, IoT devices have to log their data in the blockchain multiple times, depending on the data rate. Moreover, energy constraints may limit the lifespan of IoT devices in wireless sensor networks. These factors motivate the study of multiple mining time selections for IoT devices. We formulate a multiple stopping time Bayesian sequential detection problem as a partially observed Markov decsion process (POMDP)~\cite{2016:VK} to address energy-efficient blockchain mining for IoT devices. The objective is to identify $\numStops$ optimal stops for mining, thereby maximizing the probability of successfully adding a block to the blockchain. We assume that the dynamics of the POMDP are not affected by the mining activity of the IoT device. This assumption is reasonable since the computing power of the IoT device is too small to affect the overall rate of new blocks in the blockchain. Multiple-stopping time POMDP has not been studied much in the literature. Compared to the single-stopping time problem, the multiple-stopping time problem is a more complex generalization because using the single-stopping policy repeatedly results in a suboptimal solution.
    
    Using the optimal policy of the resulting POMDP ensures that the IoT device does not waste energy on mining when there are several miners competing to mine a new block. We also formulate an optimization problem to optimize the number of mining instants for a blockchain-enabled IoT device. This is important in IoT applications such as wireless sensor networks, as each IoT device senses data at a different rate. Hence, the amount of data that needs to be logged in the blockchain varies with IoT devices. 

    \subsection*{Submodular Structure of the Energy-Efficient Mining Problem}
    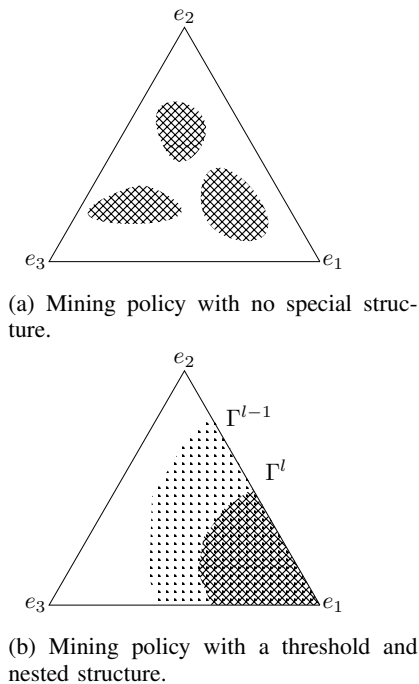
\begin{figure}
        \centering
        \begin{subfigure}{0.3\textwidth}
        \scalebox{0.9}{
        \begin{tikzpicture}
            \tikzstyle{arrow} = [draw, -];
            \tikzstyle{arrowmark} = [draw, ->];
            \tikzstyle{arrow1} = [draw, -{Implies},double];
            \def\scale{4};
            \coordinate (e1) at (1*\scale,0);
            \coordinate (e2) at (0.5*\scale,0.866*\scale);
            \coordinate (e3) at (0,0);
            \node at (1.05*\scale,0) {$e_1$};
            \node at (0.5*\scale,0.906*\scale) {$e_2$};
            \node at (-0.2,0) {$e_3$};
            \draw (e1) -- (e2) -- (e3) -- (e1);
            \fill [pattern=crosshatch, pattern color=black] plot [smooth cycle, tension=1, scale=0.4, shift={(0.4*\scale,0.4*\scale)}] coordinates {(0,0) (1,1) (2.5,1) (3,0)};
            \fill [pattern=crosshatch, pattern color=black] plot [smooth cycle, tension=1, scale=0.4, shift={(1.4*\scale,0.6*\scale)}] coordinates {(0,0) (1,1) (2.5,-1) (1,-1.5)};
            \fill [pattern=crosshatch, pattern color=black] plot [smooth cycle, tension=1, scale=0.4, shift={(0.8*\scale,1.2*\scale)}] coordinates {(1,1) (2.5,0.5) (2,-1) (1,-0.5)};
        \end{tikzpicture}
        }
        \caption{Mining policy with no special structure.}
        \label{fig:threshold-policy-arbitrary}
        \end{subfigure}
        \hspace{1cm}
        \begin{subfigure}{0.3\textwidth}
        \scalebox{0.9}{
        \begin{tikzpicture}
            \pgfdeclarepatternformonly{thick dots}{\pgfpointorigin}{\pgfqpoint{1mm}{1mm}}{\pgfqpoint{1mm}{1mm}}%
            {
            \pgfpathcircle{\pgfpointorigin}{1pt}
            \pgfusepath{fill}
            }
            \tikzstyle{arrow} = [draw, -];
            \tikzstyle{arrowmark} = [draw, ->];
            \tikzstyle{arrow1} = [draw, -{Implies},double];
            \def\scale{4}
            \coordinate (e1) at (1*\scale,0);
            \coordinate (e2) at (0.5*\scale,0.866*\scale);
            \coordinate (e3) at (0,0);
            \node at (1.05*\scale,0) {$e_1$};
            \node at (0.5*\scale,0.906*\scale) {$e_2$};
            \node at (-0.2,0) {$e_3$};
            \draw (e1) -- (e2) -- (e3) -- (e1);
            \fill [pattern=crosshatch, pattern color=black] plot [tension=1] coordinates {
            (0.6*\scale,0) (0.55*\scale,0.12*\scale) (0.57*\scale,0.25*\scale) (0.65*\scale,0.36*\scale) (0.75*\scale,0.43*\scale) (1*\scale,0) (0.6*\scale,0)};
            \node[black] at (0.84*\scale,0.5*\scale) {$\Tau^{l}$};
            \fill [pattern=thick dots, pattern color=black] plot [tension=1] coordinates {
            (0.4*\scale,0)  (0.37*\scale,0.1*\scale)  (0.38*\scale,0.4*\scale) (0.5*\scale,0.6*\scale) (0.6*\scale,0.7*\scale) (1*\scale,0) (0.4*\scale,0)};
            \node[black] at (0.74*\scale,0.7*\scale) {$\Tau^{l-1}$};
        \end{tikzpicture}
        }
        \caption{Mining policy with a threshold and nested structure.}
        \label{fig:threshold-policy-structure}
        \end{subfigure}
        \caption{Visual illustration of the structure of a mining policy. The triangle represents the two-dimensional belief space for the POMDP. The shaded regions indicate the belief state where it is optimal to mine in the blockchain. In general, the optimal mining states are arbitrary, as in (a) and computing the shaded regions is intractable. The main contribution of this paper is to propose sufficient conditions for submodularity so that the optimal policy has the threshold and nested structure as in (b). Here, $\Tau_l$ denotes the threshold for the $l^{th}$ mining instant (see Sec.\ref{sec:structural-results} for details). This structure is then exploited to develop  policy gradient algorithms.} 
    \end{figure}
    In general, solving a multiple-stopping time POMDP is P-SPACE hard~\cite{1996:DB-JT} as the optimal mining policy may not have a special structure as shown (Fig.~\ref{fig:threshold-policy-arbitrary}). Hence, often POMDPs are solved via heuristics. In this paper, 
    we show mathematically that the optimal policy for the multiple stopping time POMDP has a special structure: 1) optimal mining policy is monotone in belief space, 2) optimal mining policy has a threshold structure, thereby partitioning the belief space into two connected sets ( Fig.\ref{fig:threshold-policy-structure}). These structural results are proved via submodularity arguments on the stochastic dynamic programming equation of the POMDP. The important practical consequence is that this structure can be exploited to design policies which are linear in belief state and efficiently implementable on IoT devices. 
    The structure also facilitates the development of computationally-efficient policy gradient algorithms that can be implemented on IoT devices. Firstly, we establish both the necessary and sufficient conditions that the parameters of the linear mining policy must meet in order to satisfy the structural results. Subsequently, we convert these parameters into spherical coordinates, enabling us to formulate an unconstrained optimization problem for optimizing the parameters.
    
    \subsection*{Related works}
    The benefits of integrating blockchain and IoT is discussed in \cite{2023:EK}. \cite{2018:ON} describes an architecture for integrating IoT and blockchain, and \cite{2022:AA-et-al} proposes a medium access control (MAC) protocol for IoT-blockchain network setup.
    \cite{2021:SM-et-al} conducts simulations, and \cite{2020:NK-et-al} uses a prototype implementation to study the overall system performance while integrating blockchain with IoT.    
    
    Related to the modeling of the blockchain system, \cite{2022:XS-et-al} and \cite{2020:YL-et-al} employ a Markov process model to study performance and network security in a distributed ledger.
    The evolution of cryptocurrency as a Hidden Markov Model is explored in~\cite{2022:XS-et-al}, \cite{2022:KK-et-al}.
    
    Regarding optimal mining strategies in blockchain, 
    \cite{2020:RS-et-al} formulates the mining problem with resource cost as a dynamic game over an infinite horizon and shows that it is optimal to mine together.
    \cite{2020:GY-et-al}, \cite{2021:TW-et-al} uses reinforcement learning to optimize selfish mining strategy in the blockchain. 
    \cite{2016:AK-et-al} utilizes a game theoretic approach to study Nash equilibrium in blockchain mining when miners can hide their newly mined nodes.
    \cite{2023:YZ-et-al} formulates various selfish mining strategy in blockchain as a Markov decision process and solves it to obtain a lower bound on their performance.
    \cite{2011:MR} analyzes the optimal mining time for pooled mining reward systems. The author argues that the optimal mining time depends on the miners' incentives and the network's transaction volume.

    \cite{2021:WM-et-al} surveys the problem of energy overhead in integrating IoT and blockchain both from computational and communication viewpoints.   \cite{2019:CS-et-al} presents a clustering method to improve energy efficiency for blockchain mining in an IoT application.
    
    Our energy-efficient mining problem for a blockchain-enabled IoT device utilizes tools from \cite{2018:VK-et-al}, which provides useful structural results on the optimal policy for multiple stopping time POMDP. The structural result allows solving the POMDP for large state space. Another approach to reducing the complexity of MDP with large state space is presented in~\cite{2023:TB-UM}: the authors approximate the value function by projecting it into a lower-dimensional subspace. Also, one can use stochastic dominance to compute bounds on HMM filter with reduced computational complexity~\cite{2014:VK-RJ}. Multiple stopping time POMDP has been used for targeting ads, intrusion prevention~\cite{2022:KH-RS}, active sensing~\cite{2017:DS-UM}, sensor scheduling~\cite{2011:GA-et-al} and detecting line outage in power systems~\cite{2017:GR-et-al}. To the best of our knowledge, the multiple stopping time POMDP approach to study energy-efficient mining strategy in blockchain-enabled IoT devices has not been explored in the literature. The approach provides a computationally-efficient way to maximize energy efficiency for a blockchain-enabled IoT device. 
    
    \label{sec:introduction}
    \subsection*{Organization and Main Results}
     Sec.\ref{sec:problem-statement} describes the energy-efficient mining problem in blockchain for IoT applications; it is formulated as a multiple stopping time POMDP in Sec.\ref{sec:model}. We explore the comparison between blockchain-enabled IoT devices and multiple stopping time POMDP in Sec.\ref{sec:discussion-model}. Sec.\ref{sec:assumptions} discusses our model assumptions, and in Sec.\ref{sec:main-results}, we derive structural results on the optimal policy for the energy-efficient mining problem in blockchain. Finally, using the structural results, Sec.\ref{sec:simulations} solves for an optimal linear mining policy for a blockchain-enabled IoT device to maximize energy efficiency on a synthetic and real Bitcoin dataset. We also compare the optimal mining policy with other heuristic mining strategies. 
    	\label{sec:organization}
    	
    \section{Energy-Efficient Mining in Blockchain}
    \label{sec:problem-statement}
    \begin{figure}
        \centering
        \scalebox{0.9}{
        \begin{tikzpicture}
            \tikzstyle{arrow} = [draw, -];
            \tikzstyle{arrowmark} = [draw, ->];
            \tikzstyle{arrow1} = [draw, -{Implies},double];
            \node (blockchain) at (-0.7,0) [draw, rectangle, align=center] {\includegraphics[scale=0.42]{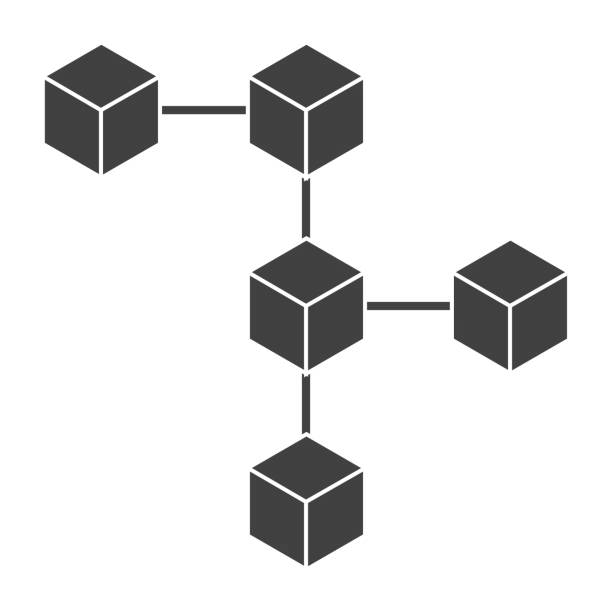}};
            \node (blockchain-text) at (-0.7,-1.6) [text width=3cm, align=center] {Blockchain};
            \node (pow) at (-0.7,2.5) [draw, rectangle, text width=1.5cm, align=center] {PoW\\difficulty\\ level};          
            \node (observable) at (-0.7,0) [draw, rectangle, dashed, text width=2.5cm, minimum height=6cm, align=center] {\vspace{6.5cm}\\ \textbf{Observable}};
            \draw [arrow1] (blockchain)--(pow);                 
            \foreach \x/\y in {{1/(2,2.5)},{2/(2,0.8)}}
            {
                \node (miner\x) at \y  {\includegraphics[scale=0.2]{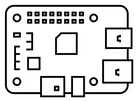}};
                \draw [arrow] (blockchain)--(miner\x);
            }
            \foreach \x/\y in {{3/(2,-2.5)},{4/(2,-0.8)}}
            {
                \node (miner\x) at \y  {\includegraphics[scale=0.13]{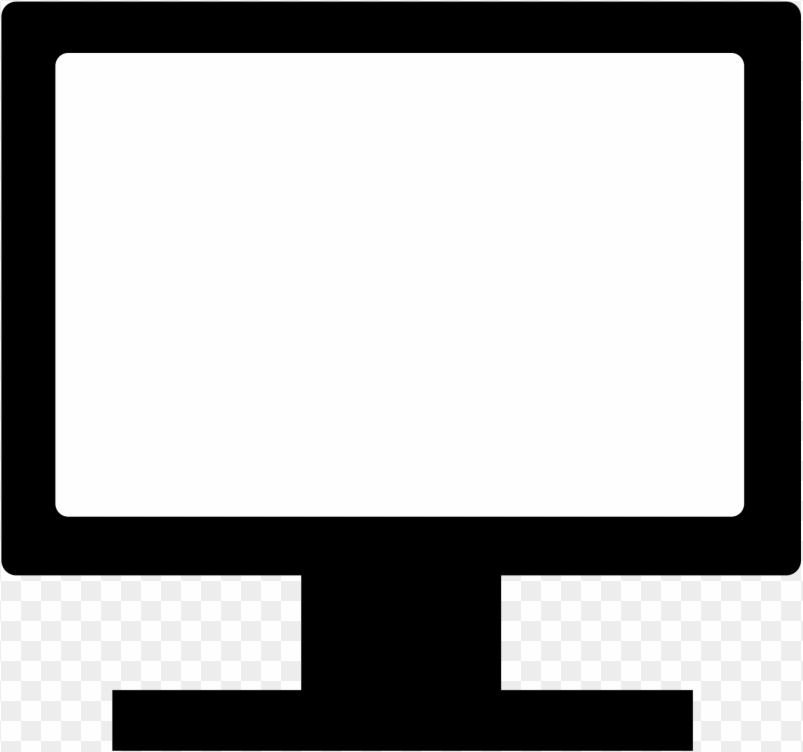}};
                \draw [arrow] (blockchain)--(miner\x);
            }
            \node (miner-dm) at (-3.5,0) {\includegraphics[scale=0.2]{raspberrypi.png}};
            \node (miner-dm-text) at (-3.5,-1.2) [text width=3cm, align=center] {IoT device\\(Decision maker)};
            \draw [arrow] (blockchain)--(miner-dm);
            \node (unobservable) at (2.1,0) [draw, rectangle, dashed, text width=2.2cm, minimum height=6cm, align=center] {\vspace{6.5cm}\\ \textbf{Unobservable}};
            \node (low-power) at (-3.6,3) {\includegraphics[scale=0.2]{raspberrypi.png}};
            \node (low-power-text) at (-3.6,2.3) [text width=5cm, align=center] {Low power devices};
            \node (high-power) at (-3.6,1.5) {\includegraphics[scale=0.13]{pc.png}};
            \node (high-power-text) at (-3.6,0.8) [text width=5cm, align=center] {High power devices};
        \end{tikzpicture}
        }
        \caption{Schematic of the energy-efficient mining problem for blockchain-enabled IoT applications with heterogeneous devices. The IoT device (decision maker) wants to decide when to mine in the blockchain to maximize its probability of adding a new block. The IoT device cannot observe the computing power of other miners, nor can it observe whether they are mining or not. It can only observe the PoW difficulty level of the blockchain, a noisy observation of the other miners' actions. Hence, the IoT device has to decide its action based only on the available information to maximize its probability of adding a new block to the blockchain.}
        \label{fig:block-diagram}
    \end{figure}
    Recall that our aim is to determine the optimal mining times for an IoT device to act as a blockchain miner.
    The purpose is to minimize the energy consumed by low-power IoT devices that log their data into a secure (tamper-proof) distributed ledger.
    
    We consider a blockchain-based distributed ledger for IoT applications, where multiple miners compete to mine the next block (see Fig.~\ref{fig:block-diagram}). Mining involves solving a cryptographic puzzle to satisfy PoW: a consensus algorithm for blockchain. An important parameter associated with PoW is the PoW difficulty level; it determines the complexity of the cryptographic puzzle to be solved to add subsequent blocks to the blockchain. The blockchain protocol adjusts the PoW difficulty level to regulate the rate of new blocks: when many miners participate in mining, the total computing power invested in the blockchain is significant, hence, decreasing the expected mining time for the next block; this ensues increase in the PoW difficulty level. As each miner has different computing power and incentives, they individually decide whether or not to invest their computing power for mining. Deciding the mining time instants is crucial for integrating blockchain in an IoT application, as IoT devices have limited energy resources. 
    
    We consider the energy-efficient mining problem from the perspective of a single miner. We model the evolution of blockchain and decisions made by the miner as a discrete-time POMDP. Specifically, the energy-efficient mining problem in the blockchain is formulated as a multiple-stopping time POMDP. The reward function encodes the miner's probability of solving the PoW puzzle faster than all other miners. The miner aims to optimize its mining policy to maximize its probability of adding the next block to the blockchain, thereby minimizing the wastage of energy resources. 
    
    Sec.\ref{sec:model} formulates the energy-efficient mining in blockchain for IoT devices as a discrete-time multiple stopping time POMDP. We explore the comparison between blockchain-enabled IoT and optimal stopping time POMDP in Sec.\ref{sec:discussion-model}. 
    
    \subsection{POMDP model for the energy-efficient mining problem in blockchain}
    \label{sec:model}
    In this section, we formulate the energy-efficient mining problem in blockchain as a multiple-stopping time POMDP. We assume that the dynamics of the POMDP are not affected by the mining activity of the IoT device. This assumption is reasonable since the computing power of the IoT device is too small to affect the overall rate of new blocks in the blockchain. We discuss this assumption in detail in~\ref{assumption:c-small}.
    Let $t=0,1,2,\ldots$ denote discrete time.
    \subsubsection{System state $\compPowerRV_t$ and the initial state distribution $\compPowerDistribution_0$} The system state $\compPowerRV_t\in\compPowerSet$ denotes the total computing power invested by all the miners in the blockchain at time $t$ with the initial distribution denoted by the pmf $\compPowerDistribution_0\in\R^{\compPowerSetSize}$. Here, $\compPowerSet=\{1,2,\ldots,\compPowerSetSize\}$ denotes the set of all possible system states. When the total computing power $\compPowerRV_t$ is large, then the mining activity in the blockchain is also large. As discussed in~\eqref{eq:obsv-matrix}, an individual device can only observe the total computing power $\compPowerRV_t$ in noise. Hence, determining the optimal mining time is non-trivial. 
    
    \subsubsection{Transition matrix $\transMatrix$}
    We model the evolution of the total computing power in the blockchain as a time-homogeneous Markov chain with transition matrix $\transMatrix$. This Markov assumption is widely used~\cite{2022:KK-et-al}. We will justify the Markov assumption in Sec.\ref{sec:simulations} using a real Bitcoin dataset. 
    
    For $i,j\in\compPowerSet$, elements of the transition matrix~$\transMatrix\in\R^{\compPowerSetSize\times\compPowerSetSize}$ are
    \begin{align}
        \label{eq:trans-matrix}
        \transMatrix(i, j)=\Prob\bracketRound{\compPowerRV_{t+1}=j \mid \compPowerRV_t=i}
    \end{align}        
    \subsubsection{Observation $\obsvRV_t$}
    An individual IoT device is unaware of the total mining activity at each time instant. Therefore, the IoT device cannot observe the total computing power $\compPowerRV_t$ invested in the blockchain at time $t$. Instead, the IoT device observes the PoW difficulty level $\obsvRV_t\in\obsvSet=\{1,2,\ldots,\obsvSetSize\}$, which can be viewed as a noisy measurement  of the total computing power $\compPowerRV_t$. The relationship between $\compPowerRV_t$ and $\obsvRV_t$ is described by the distribution $\obsvMatrix$.
    \begin{align}
        \label{eq:obsv-matrix}
        \obsvMatrix(i,\obsv)=\Prob(\obsvRV_t=\obsv \mid \compPowerRV=i),\;i\in\compPowerSet,\obsv\in\obsvSet
    \end{align}
    \subsubsection{Policy $\policy$}
    The IoT device decides its mining instants in the blockchain using the policy $\policy$. The IoT device decides whether or not to mine at time $t$ as a function of $\info_t=\{\pi_0,\action_0,\obsvRV_1,\ldots,\action_{t-1},\obsvRV_t\}\in\infoSet_t$. Here, $\info_t$ denotes the history of information available at time $t$, and $\infoSet_t$ denotes the set of all possible history of information at time $t$.
    Due to limited power, the IoT device can mine at $\numStops$ time instants over the infinite time horizon. Let $\stopNumber\in\stopNumberSet$ index the mining instants. The mining policy of the IoT device is modelled as $\policy:\infoSet\times\stopNumberSet\rightarrow\scrU$. Here $\scrU=\actionSet$ denotes the action space of the IoT device. To decide the $\stopNumber^{th}$ mining time, the IoT device chooses an action at time $t$ as $\action_t=\policy(\info_t,\stopNumber)$. At time $t$, $\action_t=1$ corresponds to don't mine, and $\action_t=2$ corresponds to mine. 
    
    \subsubsection{Reward $\reward:\compPowerSet\times\scrU\rightarrow \R$}
    \label{sec:reward}
    The reward function incentivizes the IoT device to choose its mining time judiciously. It encodes the IoT device's probability of adding a block to the blockchain. An  advantage of our submodularity-based mathematical formulation is that we only require the reward function to satisfy the following general structure\footnote{Our structural results can be generalized to the case when the difference in reward $\reward(\compPowerRV,2)-\reward(\compPowerRV,1)$ is decreasing in $\compPowerRV$.}:
    \begin{align}
        \label{eq:reward}
        \reward(\compPowerRV,1)=0
        ,\;\reward(\compPowerRV,2) \text{ is decreasing in $\compPowerRV$}
    \end{align}
    Therefore, if the total computing power invested in the blockchain is $\compPowerRV$ and the IoT device decides not to mine in the blockchain, it receives a reward of $\reward(\compPowerRV,1)=0$. If the IoT device decides to mine, it earns a reward $\reward(\compPowerRV,2)$, which encodes the probability of adding a new block in the blockchain. 
    
    We now justify~\eqref{eq:reward} for blockchain-enabled IoT devices. Typically, the transaction fee is considered in the literature as the reward function~\cite{2014:NH},\cite{2017:DN}. The expected transaction fee earned depends on the probability of adding a block to the blockchain. As typical in blockchain-enabled IoT devices~\cite{2021:WM-et-al}, the devices are responsible for both generating the data and mining a new block, making it appropriate to consider the probability of adding a block as the reward. This probability is directly proportional to the computing power of the IoT device and inversely proportional to the total computing power of the blockchain\footnote{Solving PoW involves an exhaustive search over all possible values of nonce~\cite{2008:SN}. The search continues until a desired number of zeros is found at the beginning of the hash code. Therefore, a miner with higher computing power can search more nonce per unit of time, leading to a higher probability of successfully adding a block to the blockchain. Thus, the probability that an IoT device will add a block to the blockchain can be modelled as a Bernoulli random variable with $p=\frac{1}{\compPowerRV}$. This assumes that the computing power of the IoT device is normalized to 1, and it is small compared to the total computing power $\compPowerRV$.} $\compPowerRV$. 
    The assumption~\eqref{eq:reward} that the reward is decreasing emerges naturally with the above justification of the transaction fee. The assumption also enables us to characterize mathematically the optimal mining policy in Sec.\ref{sec:structural-results}. In the absence of this assumption, one would have to rely on heuristic mining policies, several of which are discussed in Sec.\ref{sec:simulations}.

    \subsubsection{Cumulative reward $\totalReward$}
    For an IoT device, choosing the mining instants myopically to maximize the reward $\reward$ is not satisfactory as that does not exploit knowledge of the Markov evolution of the state (total computing power~$\compPowerRV_t$). In this paper, we choose the mining times by maximizing the cumulative reward over a finite but random horizon resulting in a multiple-stopping time stochastic control problem. From a practical point of view, the cumulative reward represents the overall energy efficiency of the IoT device. The goal of the IoT device is to maximize its cumulative reward over $\numStops$ mining instants.  
    \begin{align}
    \label{eq:cumulative-reward}
    \begin{aligned}                         &\totalReward_\policy(\compPowerDistribution_0) = \mathbb{E}_\policy\left[\sum_{t=0}^{\stopTime_1-1}\discountFactor^{t} \reward(\compPowerRV_t,1)+\discountFactor^{\stopTime_1} \reward(\compPowerRV_{\stopTime_1},2)+\ldots
    \right.\\
    &\left.
    \ldots+\sum_{t=\stopTime_{\numStops-1}+1}^{\stopTime_\numStops-1}\discountFactor^{t} \reward(\compPowerRV_t,1)+\discountFactor^{\stopTime_\numStops} \reward(\compPowerRV_{\stopTime_\numStops},2)\mid \compPowerDistribution_0\right]\\
    &\stopTime_{i+1} = \min\{t>\stopTime_{i}:\action_t=2\},\; i\in\{1,\ldots,\numStops\},\;\stopTime_0 = 0
    \end{aligned}
    \end{align}
    Here, $\stopTime_i$ represents the $i^{th}$ mining instant in the blockchain. The discount factor $\rho$ represents the IoT device's decreased value assigned to rewards obtained in the future. One can also consider risk-averse cost function as in~\cite{2021:NU-et-al}. The optimal policy ensures that the IoT device maximizes the cumulative reward for adding a new block to the blockchain, thereby maximizing energy efficiency. 
    
    It is important to emphasize that the total computing power $\compPowerRV_t$ is not observed by the IoT device. It only observes the PoW difficulty level $\obsvRV_t$, which is a noisy measurement of the total computing power. Put simply, by observing a noisy Markov chain, what are the optimal $\numStops$ mining instants? Therefore, \eqref{eq:cumulative-reward} is a multiple-stopping time POMDP. 

    \subsection{Belief State Representation}
    \label{sec:belief-state}
    The POMDP for the energy-efficient mining problem in blockchain (Sec.\ref{sec:model}) can be formulated as the standard Markov decision process (MDP) by introducing a belief state. The belief state $\compPowerDistribution_t$ is the posterior  probability distribution of the underlying state given the  observations until the present time. It is updated recursively using the Bayesian update~\cite{2016:VK} as new observations are received.
    \begin{align}
        \label{eq:bayesian-update}
        \begin{aligned}
        \compPowerDistribution_{t+1}&=T(\compPowerDistribution_t,\obsvRV_t)\\
        T(\compPowerDistribution,\obsv)&=\frac{\obsvMatrix_y\transMatrix^\intercal\compPowerDistribution}{\sigma(\compPowerDistribution,\obsv)},\quad \sigma(\compPowerDistribution,\obsv)=\mathbb{1}_\compPowerSetSize^\intercal \obsvMatrix_\obsv\transMatrix^\intercal\compPowerDistribution
        \end{aligned}
    \end{align}
    Here, $\obsvMatrix_\obsv=\operatorname{diag}(\obsvMatrix(1,\obsv),\obsvMatrix(2,\obsv),\ldots,\obsvMatrix(\compPowerSetSize,\obsv))$ where $\obsvMatrix$ is the observation matrix~\eqref{eq:obsv-matrix}; and $\mathbb{1}_\compPowerSetSize$ represents the $\compPowerSetSize$-dimensional column vector of ones and its transpose is denoted as $\mathbb{1}_\compPowerSetSize^\intercal$.
    
    In the MDP formulation, one designs the policy $\policy$ as a function of belief state $\compPowerDistribution_t\in\Pi$. Here, $\Pi$ is a simplex, also known as belief space. It is well-known that the belief state is a sufficient statistic~\cite {2016:VK}, and designing a policy as a function of the belief state yields the same optimal solution. Using belief state facilitates analysis, but because the belief space is a simplex, it yields an MDP with continuous state-space $\Pi$. In this paper, we use the equivalent MDP formulation of the energy-efficient mining problem in blockchain using belief state to derive our structural results in Sec.\ref{sec:structural-results}. The latter is used to design a linear mining policy as a function of the belief state. Algorithm~\ref{algo:pomdp-simulator} provides the steps to compute the cumulative reward for the energy-efficient mining problem for a given mining policy $\policy$.
    \begin{algorithm}        
     \caption{Simulating energy-efficient mining problem in blockchain using belief state $\compPowerDistribution_t$ given the mining policy $\policy$}
     \begin{algorithmic}[1]
     \REQUIRE $\compPowerDistribution_0,\compPowerSet,\obsvSet,\transMatrix,\obsvMatrix,\reward,\discountFactor,\numStops$ and an upper limit on the horizon length $T$
      \STATE Initialize $\stopNumber\leftarrow 1,\;J\leftarrow 0$
      \FOR {$t = 0, 1, 2, \ldots, T$}
      \STATE Compute $\action_t\leftarrow \policy(\compPowerDistribution_t,\stopNumber)$
      \IF {($\action_t = 2$)}
        \STATE $J\leftarrow J+\rho^t\reward(\compPowerDistribution_t,2)$, $\stopNumber\leftarrow \stopNumber+1$.
          \IF {($\stopNumber = \numStops$)}
          \RETURN $J$
          \ENDIF
      \ELSE 
        \STATE $J\leftarrow J+\rho^t\reward(\compPowerDistribution_t,1)$
      \ENDIF
      \STATE Generate a new observation $\obsv_t$ and compute $\compPowerDistribution_{t+1}$ using~\eqref{eq:bayesian-update}.
      \ENDFOR
     \end{algorithmic} 
     \label{algo:pomdp-simulator}
     \end{algorithm}
	
	\subsection{Equivalent formulation as a discounted-cost POMDP}
	The multiple-stopping time POMDP for the energy-efficient mining in blockchain can be formulated as an infinite-horizon POMDP. This is achieved by augmenting a fictitious absorbing state $\compPowerSetSize+1$ with the continue reward $\reward(\compPowerSetSize+1,1)=0$. When the last stop is made, the belief state $\compPowerDistribution_t$ (ref. Sec.\ref{sec:belief-state}) transitions to $e_{\compPowerSetSize+1}$. Here, $e_{\compPowerSetSize+1}=(0,\ldots,0,1)\in\R^{\compPowerSetSize+1}$. The cumulative reward~\eqref{eq:cumulative-reward} for the multiple stopping POMDP is equivalent to $\totalReward_\policy(\compPowerDistribution_0) = \mathbb{E}_\policy\left[\sum_{t=0}^{\stopTime_1-1}\discountFactor^{t} \reward(\compPowerRV_t,1)+\discountFactor^{\stopTime_1} \reward(\compPowerRV_{\stopTime_1},2)+
		\ldots\right.$ $\left.+\sum_{t=\stopTime_{\numStops-1}+1}^{\stopTime_\numStops-1}\discountFactor^{t} \reward(\compPowerRV_t,1)+\discountFactor^{\stopTime_\numStops} \reward(\compPowerRV_{\stopTime_\numStops},2)+\right.\\\left.\sum_{t=\stopTime_{\numStops}+1}^{\infty}\discountFactor^{t} \reward(\compPowerSetSize+1,1)\mid \compPowerDistribution_0\right]$. In the standard form of POMDP, the transition matrix depends on the input. To obtain an input-dependent transition matrix, one can use the modified state $(\stopNumber,\compPowerRV),\compPowerRV\in\compPowerSet,\stopNumber\in\stopNumberSet$. Here, $\stopNumber$ denotes the stop number and $\compPowerRV$ denotes the original state (ref. Sec.\ref{sec:model}). To specify the new transition matrix, we order the modified state as:
		\begin{align*} ((1,1),\ldots,(1,\compPowerSetSize),(2,1),\ldots,(2,\compPowerSetSize),\ldots,(\numStops,1),\ldots,\\(\numStops,\compPowerSetSize),\compPowerSetSize+1))\end{align*} 
		For this ordering, the transition matrix with $\action=1$~(ref.~\eqref{eq:reward}) is given by $\widetilde{\transMatrix}_{\action=1}=\operatorname{diag}\{\transMatrix,\ldots,\transMatrix,1\}$. Here, operator $\operatorname{diag}$ is used to construct a block diagonal matrix, and $\transMatrix$ is defined in~\eqref{eq:trans-matrix}. The transition matrix with $\action=2$~(ref.~\eqref{eq:reward}) is given by $\widetilde{\transMatrix}_{\action=2}=\begin{bmatrix}
		0 &\transMatrix & & &  \\
		0 & 0 &\transMatrix & &\\
		\vdots & & \ddots & &\vdots\\
		0 & \cdots & 0 &\transMatrix &0\\
		0 & \cdots & 0 & 0 & 1_{\compPowerSetSize\times 1}\\
		0 & \cdots & 0 & 0 & 1
		\end{bmatrix}$
    \subsection{Discussion. Multiple Stopping Time Model for Blockchain-Enabled IoT}
    \label{sec:discussion-model}
    We now discuss the POMDP model in the context of blockchain for IoT.
    The combination of IoT and blockchain enables secure, transparent and scalable data sharing amongst a large number of users~\cite{2023:EK}\cite{2018:DF-KM}. We consider the example of a sensor network consisting of heterogeneous IoT devices (see Fig.\ref{fig:block-diagram}): low-power devices like Raspberry Pi and high-power devices like PCs. Low-power devices are typically used for data collection at remote locations, while high-power devices monitor time-critical applications. In our setup, the devices in the sensor network use a blockchain platform, like Etherium~\cite{2017:CD}, to log their data. The IoT devices have to compete to solve a PoW faster than other miners to add a new block in the blockchain. This improves the security of the blockchain: it is difficult to tamper with transactions in the blockchain. However, PoW involves solving a cryptographic puzzle which is energy-intensive. As low-power IoT devices have limited energy resources, they need to optimize their mining time instants to maximize their probability of adding a new block to the blockchain. This would minimize the wastage of energy by resource-constrained IoT devices. We described our POMDP formulation for the energy-efficient mining problem in blockchain for IoT applications in Sec.\ref{sec:model}. We now discuss the model parameters in the context of IoT and blockchain.
    \subsubsection{Markovian system dynamics for mining  in blockchain}
    The probability of a miner adding a new block to the blockchain is determined by the total computing power invested in it. Therefore, it is crucial for an IoT device to keep track of the total computing power. Our approach involves modeling the total computing power invested in the blockchain as a Markov chain with a transition matrix P. This matrix captures how the computing power changes over time as individual miners make decisions based on their own trade-offs between mining cost and reward. 
    
    \textit{Remarks. Estimating the transition matrix. } Even though the actual total computing power in the blockchain is not directly observable, it can still be estimated based on the rate of new blocks and the PoW difficulty level. \cite{2023:blockchain} includes a record of the estimated total computing power in the past, which can be used to estimate the transition matrix. This can be done by grouping the historical data into a specified number of states and applying a maximum likelihood estimator (MLE) to obtain the transition matrix.
    \subsubsection{PoW difficulty level as a noisy observation of the system state}
    Recall that the observations are the PoW difficulty level, a noisy observation of the system state.
    With a larger total computing power, new blocks are mined faster on average, and thus, the blockchain protocol adjusts the PoW difficulty level to maintain a constant rate of new blocks. As the IoT device cannot observe the total computing power invested in the blockchain, it
    uses the PoW difficulty level to update the belief about the total computing power.
    
    \textit{Remarks. Estimating the observation distribution. } \cite{2023:blockchain} provides a historical record of the estimated total computing power and the PoW difficulty level in the blockchain. One can use an MLE to estimate the observation distribution $\obsvMatrix$ from the data. 
    \subsubsection{Probability of adding a new block in the blockchain}
    The IoT device wants to mine the blockchain to log its sensor readings while maximizing its energy efficiency, which is defined as the probability of adding the next block to the blockchain. This preference is modelled as the reward in the optimal stopping time problem within the framework of POMDP. The reward function is proportional to the computing power of the IoT device and inversely proportional to the total computing power invested in the blockchain if the IoT device decides to mine. Otherwise, it receives no reward.
    \subsubsection{Mining policy and the total number of mining instants}
    We model a single IoT device as a decision maker and optimize its mining time instants so as to maximize its energy efficiency. Due to the energy constraint imposed on the IoT device, it can only engage in mining for a finite number of time instants. Consequently, the IoT device seeks to increase its likelihood of adding a new block to the blockchain in order to minimize energy wastage. 
    \subsection{Optimizing the number of mining instants in blockchain}
    In Sec.\ref{sec:model}, we presented our model for the energy-efficient mining problem in the blockchain. Our model assumed that the number of mining instants $\numStops$ was fixed and known to the IoT device. In a realistic scenario, the IoT device also has to optimize the number of mining instants $\numStops$. This is because different sensors (IoT devices) in an IoT application record data at different rates based on their task. The amount of data that needs to be logged in the blockchain is proportional to the data rate. 
    The optimization problem to optimize the number of mining instants $\numStops$ is:
    \begin{align}
    \label{eq:optimize-numstops}
            \max_{\numStops}\totalReward_{\policy^*,\numStops}(\compPowerDistribution_0)-\numStopsCost(\numStops)
    \end{align}
    Here, $\totalReward_{\policy^*,\numStops}(\compPowerDistribution_0)$ denotes the optimal cumulative reward~\eqref{eq:cumulative-reward} when number of mining instants is $\numStops$; $\numStopsCost$ is the cost function for choosing a particular $\numStops$. We assume that $\numStopsCost$ is an increasing and convex function of $\numStops$. This is because the energy consumed increases with the number of mining instants $\numStops$. Also, as the energy consumption increases, the size of the battery increases and incurs additional cost; this motivates the convexity of $\numStopsCost$. It can be empirically verified that $\totalReward_{\policy^*,\numStops}$ is concave in $\numStops$. Therefore, \eqref{eq:optimize-numstops} is a convex optimization problem. Although $\numStops$ is discrete-valued, we can solve~\eqref{eq:optimize-numstops} in continuous domain and compare the nearest integer solutions to obtain the optimal number of mining instants.
    
     To summarize, we formulated the energy-efficient mining problem in blockchain as a multiple-stopping time POMDP. Due to the curse of dimensionality, it is difficult to solve the optimal mining policy. So, in Sec.\ref{sec:structural-results}, we derive structural results, which would be exploited in Sec.\ref{sec:simulations} for obtaining an optimal linear mining policy.
    \section{Structural Results for Energy-Efficient Mining in Blockchain}
    \label{sec:structural-results}
    This section presents structural results for the energy-efficient mining problem in blockchain (ref. Sec.\ref{sec:problem-statement}).
    Sec.\ref{sec:assumptions} discussed the model assumptions to derive the structural results. In Sec.\ref{sec:main-results}, we first show that the optimal mining policy in blockchain has a threshold structure (Theorem~\ref{thm:monotonicity}). Sec.\ref{sec:implications} discusses the significance of the structural results for blockchain-enabled IoT devices. Sec.\ref{sec:linear-policy} exploits Theorem~\ref{thm:monotonicity} to design a linear mining policy. This is followed by necessary and sufficient conditions on the parameters (Theorem~\ref{thm:parameter-condition}) of the linear mining policy to satisfy the structural results. Optimizing the parameters of the linear policy corresponds to solving a constrained optimization problem which is difficult. Hence, Sec.\ref{sec:spherical-coordinates} describes the parameters of the linear mining policy in spherical coordinates. The spherical coordinates simplify the problem of optimizing the parameters to an unconstrained optimization problem. Sec.\ref{sec:spsa} discusses the policy gradient algorithm to optimize the parameters of the linear mining policy in spherical coordinates. The main outcome of this section is to construct a computationally efficient, optimal linear mining policy for the blockchain-enabled IoT device. This is achieved by exploiting the structural results for the selection of the policy parameters.

    \subsection{Assumptions on multiple-stopping time POMDP model} 
    \label{sec:assumptions}
    We now discuss our assumptions on the model for the energy-efficient mining problem in blockchain for IoT applications. To understand our assumptions, we need to define the property of total positivity of order 2 (TP2).
    \begin{definition}[Total positivity of order 2 (TP2) \cite{2016:VK}]
    \label{def:tp2}
        A stochastic matrix $A$ is TP2 if all the second-order minors are non-negative, i.e., the determinants $\begin{vmatrix}
            A_{i_1j_1} & A_{i_1j_2}\\
            A_{i_2j_1} & A_{i_2j_2}
            \end{vmatrix}\geq 0,\;\forall i_1<i_2, j_1<j_2$
        .Here, $A_{ij}$ denotes the $(i,j)^{th}$ element of the matrix $A$.
    \end{definition}
    The assumptions~\ref{assumption:slow-mc}-\ref{assumption:c-small} will serve as the basis for deriving the structural properties of the optimal policy $\policy^*$ in Sec.\ref{sec:structural-results}.
    \begin{enumerate}[label=(A\arabic*)]            
        \item \label{assumption:slow-mc} The transition matrix $\transMatrix$ is totally positive of order 2 (TP2)~(Definition~\ref{def:tp2}). To satisfy the TP2 assumption, we impose two conditions on the Markov chain $\compPowerRV_t$: (1) it varies slowly with time, i.e., diagonal terms are dominant, (2) the transition matrix $\transMatrix$ has a tri-diagonal structure. 

        {\em Justification:} If there are no collusions among the miners, the mining activity changes slowly with time. Hence, the tri-diagonal assumption is valid, and it can be satisfied by using a small enough sampling time and/or appropriately binning the states of the Markov chain $\compPowerRV_t$. 
        \item \label{assumption:obsv-tp2} Observation distribution $\obsvMatrix$ (defined in~\eqref{eq:obsv-matrix}) is totally positive of order~2~(TP2). 
        
        \textit{Justification:} Since the observations are the PoW difficulty level, they are non-negative integers. We can employ empirical methods to approximate the observation distribution within the class of TP2 distributions. This is not restrictive as several well-known distributions over non-negative integers satisfy this property~\cite{2002:AM-DS} such as binomial, Poisson, geometric distribution, etc. In the numerical section (Sec.\ref{sec:simulations}) involving a real Bitcoin dataset, we fit the data to the nearest TP2 distribution to the observation distribution.
        \item \label{assumption:c-small} The dynamics of $\compPowerRV_t$ are not affected by the miner's (decision maker) action $\action_t$ at time $t$. 
        
        \textit{Justification:} This assumption is realistic for an IoT device with small computing power compared to the total computing power (\cite{2015:WA}\cite{2021:WM-et-al} provides a comparative study of the typical computing power and energy used by an IoT device and a PC). This assumption is further justified because the IoT device's computing power is negligible to make a significant impact on the rate at which new blocks are added in the blockchain\footnote{\ref{assumption:c-small} allows us to deploy our model for multiple low-power IoT devices as long as the total computing power of the IoT devices is significantly small compared to the total computing power invested in the blockchain.}. 
    \end{enumerate}  

    \subsection{Structural results for the optimal mining policy}    
    \label{sec:main-results}
    \begin{figure}
        \centering
        \begin{subfigure}{0.25\textwidth}
        \scalebox{0.9}{
        \begin{tikzpicture}
            \tikzstyle{arrow} = [draw, -];
            \tikzstyle{arrowmark} = [draw, ->];
            \tikzstyle{arrow1} = [draw, -{Implies},double];
            \def\scale{4}
            \coordinate (e1) at (1*\scale,0);
            \coordinate (e2) at (0.5*\scale,0.866*\scale);
            \coordinate (e3) at (0,0);
            \node at (1.05*\scale,0) {$e_1$};
            \node at (0.5*\scale,0.906*\scale) {$e_2$};
            \node at (-0.2,0) {$e_3$};
            \draw (e1) -- (e2) -- (e3) -- (e1);
            \coordinate (pi1) at (0.5,0.866);
            \coordinate (pi2) at (0.5*2,0.866*2);
            \coordinate (pi3) at (0.5*3,0.866*3);
            \draw[dash dot] (e1) -- (pi1) node[font=\small, near end, above, sloped] {$\scrL(e_1,\pi_1)$};
            \draw[dash dot] (e1) -- (pi2) node[font=\small, near end, above, sloped] {$\scrL(e_1,\pi_2)$};
            \draw[dash dot] (e1) -- (pi3) node[font=\small, near end, above, sloped] {$\scrL(e_1,\pi_3)$};
            \node at (0.26,0.866) {$\bar{\pi_1}$};
            \node at (0.38*2,0.866*2) {$\bar{\pi_2}$};
            \node at (0.42*3,0.866*3) {$\bar{\pi_3}$};
            \draw[green!50!black] (0.6*\scale,0) -- (0.55*\scale,0.12*\scale);
            \draw[green!50!black] (0.55*\scale,0.12*\scale) -- (0.57*\scale,0.25*\scale);
            \draw[green!50!black] (0.57*\scale,0.25*\scale) -- (0.65*\scale,0.36*\scale);
            \draw[green!50!black] (0.65*\scale,0.36*\scale) -- (0.75*\scale,0.43*\scale);
            \node[green!50!black] at (0.84*\scale,0.5*\scale) {$\Tau^{l-1}$};
            \draw[blue] (0.7*\scale,0) -- (0.65*\scale,0.09*\scale);
            \draw[blue] (0.65*\scale,0.09*\scale) -- (0.66*\scale,0.2*\scale);
            \draw[blue] (0.66*\scale,0.2*\scale) -- (0.72*\scale,0.29*\scale);
            \draw[blue] (0.72*\scale,0.29*\scale) -- (0.8*\scale,0.35*\scale);            
            \node[blue] at (0.88*\scale,0.35*\scale) {$\Tau^{l}$};
        \end{tikzpicture}
        }
        \caption{General mining policy with a threshold structure}
        \label{fig:threshold-policy-general}
        \end{subfigure}
        \hspace{1cm}
        \begin{subfigure}{0.25\textwidth}
        \scalebox{0.9}{
        \begin{tikzpicture}
            \tikzstyle{arrow} = [draw, -];
            \tikzstyle{arrowmark} = [draw, ->];
            \tikzstyle{arrow1} = [draw, -{Implies},double];
            \def\scale{4}
            \coordinate (e1) at (1*\scale,0);
            \coordinate (e2) at (0.5*\scale,0.866*\scale);
            \coordinate (e3) at (0,0);
            \node at (1.05*\scale,0) {$e_1$};
            \node at (0.5*\scale,0.906*\scale) {$e_2$};
            \node at (-0.2,0) {$e_3$};
            \draw (e1) -- (e2) -- (e3) -- (e1);
            \coordinate (pi1) at (0.5,0.866);                
            \coordinate (pi2) at (0.5*2,0.866*2);
            \coordinate (pi3) at (0.5*3,0.866*3);
            \node at (0.26,0.866) {$\bar{\pi_1}$};
            \node at (0.38*2,0.866*2) {$\bar{\pi_2}$};
            \node at (0.42*3,0.866*3) {$\bar{\pi_3}$};
            \draw[dash dot] (e1) -- (pi1) node[font=\small, near end, above, sloped] {$\scrL(e_1,\pi_1)$};
            \draw[dash dot] (e1) -- (pi2) node[font=\small, near end, above, sloped] {$\scrL(e_1,\pi_2)$};
            \draw[dash dot] (e1) -- (pi3) node[font=\small, near end, above, sloped] {$\scrL(e_1,\pi_3)$};
            \draw[green!50!black] (0.6*\scale,0) -- (0.75*\scale,0.43*\scale);
            \node[green!50!black] at (0.84*\scale,0.5*\scale) {$\Tau^{l-1}$};
            \draw[blue] (0.7*\scale,0) -- (0.8*\scale,0.35*\scale);
            \node[blue] at (0.88*\scale,0.35*\scale) {$\Tau^{l}$};
        \end{tikzpicture}
        }
        \caption{Linear mining policy with a threshold structure}
        \label{fig:threshold-policy-linear}
        \end{subfigure}
        
        \caption{Visual illustration of Theorem~\ref{thm:monotonicity} and Theorem~\ref{thm:parameter-condition} for $\compPowerSetSize=3$. $\Tau^{l}$ denotes the threshold for deciding the $l^{th}$ mining instant. $\Tau^{l}$ partitions the belief space $\Pi$ into two connected regions $M^l$ (right of $\Tau^l$) and $D^l$ (left of $\Tau^l$). The linear policy in Sec.\ref{sec:linear-policy} approximates the threshold $\Tau^l$ using a linear hyperplane. The dashed lines $\scrL(e_1,\pi)$  in the figure are used in the proof to show the existence of a threshold policy. From a practical point of view, we exploit the structure to estimate the optimal linear mining policy using a stochastic gradient algorithm.}
        
        \label{fig:threshold-policy}
    \end{figure}
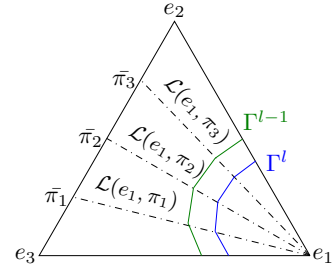
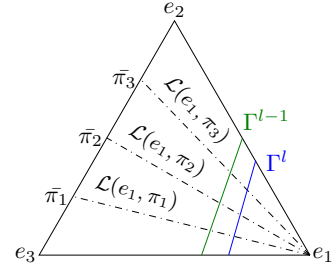
    Before discussing our structural results, we need to define the maximum likelihood ratio (MLR) ordering on the belief space $\Pi$ (see Sec.\ref{sec:belief-state}). The MLR ordering is preserved under conditional expectations~\cite{2016:VK}, making it suitable for Bayesian problems. The MLR ordering defines a partial order on a simplex, and we use it to show that the mining policy is monotone with respect to the belief state. 
    \begin{definition}[MLR ordering]
        \label{def:mlr}
        Let $\pi_1,\pi_2\in\Pi$ be two belief states. Then, $\pi_1$ is greater than $\pi_2$ with respect to MLR ordering, denoted as $\pi_1\geq_r\pi_2$, if $
            \pi_1(j)\pi_2(i)\geq\pi_2(j)\pi_1(i),\forall i<j$
    \end{definition}
    
    To understand the threshold property, let us define two families of sets: (1) \textit{mine set} $M^{\stopNumber},\;\stopNumber\in\stopNumberSet$ containing the belief states where it is optimal to mine, (2) \textit{don't mine set} $D^{\stopNumber},\;\stopNumber\in\stopNumberSet$ containing the belief states where it is optimal to not mine.
    \begin{align}
    \label{eq:space-partition}
        \begin{aligned}    
            D^\stopNumber=\{\compPowerDistribution:\policy^*(\compPowerDistribution,\stopNumber)=1\},\quad
            M^\stopNumber=\{\compPowerDistribution:\policy^*(\compPowerDistribution,\stopNumber)=2\}
        \end{aligned}
    \end{align} 
    Theorem~\ref{thm:monotonicity} shows the existence of an optimal mining policy that partitions the belief space $\compPowerDistributionSpace$ into two connected regions for each $\stopNumber\in\stopNumberSet$. Moreover, the family of sets $M^\stopNumber$ and $D^\stopNumber$ are nested. Theorem~\ref{thm:monotonicity} also shows the monotonicity of the optimal mining policy $\policy^*$ in the belief space $\compPowerDistributionSpace$. Fig.~\ref{fig:threshold-policy-general} shows a visual illustration of the Theorem~\ref{thm:monotonicity}.
    
    The belief space $\compPowerDistributionSpace$ consists of probability simplices. Hence, a total order can not be defined. To show the monotonicity of the optimal mining policy $\policy^*$, we define a family of total order subsets $\scrL(e_i,\bar{\compPowerDistribution})$ of the belief space $\compPowerDistributionSpace$. Here, $i\in\compPowerSetIndices$ indexes the elements in $\compPowerSet$ and $e_i$ denotes the $i^{th}$ standard basis vector in $\R^\compPowerSetSize$.
    \begin{align}
    \label{eq:totally-ordered-subsets}
    \begin{aligned}
        \scrH_i &\coloneqq \{\bar{\compPowerDistribution}\in\compPowerDistributionSpace,\;\bar{\compPowerDistribution}_i=0\}\\
        \scrL(e_i,\bar{\compPowerDistribution}) &\coloneqq \{\compPowerDistribution
        \mid\compPowerDistribution=\gamma e_i+(1-\gamma)\bar{\compPowerDistribution},\;\gamma\in[0,1]\},\;\bar{\compPowerDistribution}\in\scrH_i
    \end{aligned}
    \end{align}

    The set $\scrL(e_i,\bar{\compPowerDistribution}),i\in\compPowerSetIndices$ consists of line segments in the belief space $\compPowerDistributionSpace$; each element defines a totally ordered subset of the belief space $\compPowerDistributionSpace$ with respect to the monotone likelihood ratio (MLR) ordering (Definition~\ref{def:mlr}). 
In Theorem~\ref{thm:monotonicity}, we use the MLR order to show that the optimal mining strategy in blockchain is monotonically decreasing in the belief state on the lines $\scrL(e_1,\bar{\compPowerDistribution})$ and $\scrL(e_\compPowerSetSize,\bar{\compPowerDistribution})$.

    \begin{theorem}
    \label{thm:monotonicity}
    Under assumptions \ref{assumption:slow-mc}-\ref{assumption:c-small}, for each $\stopNumber\in\stopNumberSet$,
    \begin{enumerate}[label=\upshape \Alph*),ref=\thetheorem.\Alph*]
        \item\label{subthm:monotonicity} There exists an optimal policy $\policy^*(\compPowerDistribution,\stopNumber)$ that is decreasing on lines $\scrL(e_1, \bar{\compPowerDistribution})$, and $\scrL(e_S ,\bar{\compPowerDistribution})$ (defined in \eqref{eq:totally-ordered-subsets}).
        \item\label{subthm:threshold} There optimal policy $\policy^*$ partitions the belief space $\compPowerDistributionSpace$ into two individually connected sets $M^{\stopNumber}$ and $D^{\stopNumber}$ (defined in \eqref{eq:space-partition}).
        \item\label{subthm:nested} $M^{\stopNumber-1} \supset M^{\stopNumber}$
    \end{enumerate}
    \end{theorem}
    \begin{proof}
        See Sec.\ref{proof:thm-monotonicity} of supplementary material.
    \end{proof}
    Theorem~\ref{subthm:monotonicity} asserts that the optimal mining strategy $\policy^*(\compPowerDistribution,\stopNumber),\;\stopNumber\in\stopNumberSet$ is monotonically decreasing on lines $\scrL(e_1, \bar{\compPowerDistribution})$, and $\scrL(e_S ,\bar{\compPowerDistribution})$. This implies that there exists a threshold for the belief state $\pi$ above which it is optimal to mine in the blockchain. Theorem~\ref{subthm:threshold} shows that the threshold partitions the belief space into two connected sets. Theorem~\ref{subthm:nested} shows that the sets of belief state $\compPowerDistribution$, indexed by $\stopNumber\in\stopNumberSet$, such that $\policy^*(\compPowerDistribution,\stopNumber)=1$ are nested. 

    \subsection{Implications for Energy-Efficient Mining}
    \label{sec:implications}
    For modeling the energy-efficient mining problem, one needs to discretize the set of total computing power invested in the blockchain. This can lead to a large number of states for the POMDP formulation described in Sec.\ref{sec:model}. Therefore, the dynamic programming solution to POMDP is not practical for implementation on IoT devices. This is because the look-up table corresponding to the optimal mining policy grows exponentially with the number of states and requires search operations at each time instant. This is detrimental for IoT devices which have limited computational and energy resources. Theorem~\ref{thm:monotonicity} enables a less demanding approach to store the optimal mining policy in blockchain, both in terms of memory requirements and computational complexity. Memory requirement is reduced by solving a parametrized policy; this also reduces computational complexity as search operations are avoided. Moreover, under the assumption that the blockchain is time-invariant, IoT devices can be pre-programmed with the optimal linear mining policy before their deployment. This alternate solution approach will be discussed in Sec.\ref{sec:simulations}, where we describe our approach to compute an optimal linear policy for the energy-efficient mining problem in the blockchain.

    \subsection{Linear mining policy for a blockchain-enabled IoT device}
    \label{sec:linear-policy}
    This subsection focuses on the design of a linear mining policy for a blockchain-enabled IoT device which meets the structural results outlined in Theorem~\ref{thm:monotonicity}. Our main result is summarized in Theorem~\ref{thm:parameter-condition}, which characterizes the conditions on the parameters of the linear mining policy~\eqref{eq:linear-policy}.
    
    Consider a linear mining policy of the form
    \begin{align}
        \label{eq:linear-policy}
        \policy_\theta(\compPowerDistribution, \stopNumber)=\left\{\begin{array}{ll}
        2, & \begin{bmatrix}
        \theta_\stopNumber & 1 & 0
        \end{bmatrix}\left[\begin{smallmatrix}
        -1\\
        \compPowerDistribution
        \end{smallmatrix}\right] \geq 0 \\
        1, & \text { otherwise }
        \end{array}\right.
    \end{align}
    Here, $\theta_\stopNumber\in\R^{\numStops-1}$ is the parameter for the linear mining policy to decide $\stopNumber^{th}$ mining instant. $\theta$. We can restrict the search space for $\theta$ using structural results from Sec.\ref{sec:structural-results}. Theorem~\ref{thm:parameter-condition} enumerates necessary and sufficient conditions on the parameter $\theta$ so that the linear mining policy~\eqref{eq:linear-policy} satisfies Theorem~\ref{thm:monotonicity}. The conditions guarantee that all MLR-decreasing linear policies are included and no non-MLR-decreasing linear policies are excluded. Fig.~\ref{fig:threshold-policy-linear} shows a visual illustration of the linear policy~\eqref{eq:linear-policy} and Theorem~\ref{thm:parameter-condition}. 
    
    \begin{theorem}
    \label{thm:parameter-condition}
    Assuming the set $M^{\stopNumber}$ is non-empty, the necessary and sufficient conditions for the linear policy~\eqref{eq:linear-policy} to satisfy the structural results in Theorem~\ref{thm:monotonicity} are: 1) $\theta_\stopNumber(i)\geq 0,\;\forall i, \forall \stopNumber$, 2) $\theta_\stopNumber(2)\geq 1,\;\forall \stopNumber$ and $\theta_\stopNumber(i)\leq \theta_\stopNumber(2),\;\forall i>2, \forall \stopNumber$, 3) $\theta_{\stopNumber}(1) \leq \theta_{\stopNumber+1}(1)$ and $\theta_{\stopNumber}(i) \geq \theta_{\stopNumber+1}(i),\;\forall i>1, \forall \stopNumber$
    \end{theorem}
    \begin{proof}
        See Sec.\ref{proof:thm-parameter-condition} of supplementary material.
    \end{proof}
    
    \subsection{Parameters of the linear mining policy in spherical coordinates}
    \label{sec:spherical-coordinates}
    The optimization of the linear mining policy~\eqref{eq:linear-policy} subject to conditions in Theorem~\ref{thm:parameter-condition} can be formulated as a constrained optimization problem. In this subsection, we present a transformation of the policy parameters $\theta$ into spherical coordinates. This transformation will be exploited in Sec.\ref{sec:spsa} to formulate the optimization of the policy parameters as an unconstrained optimization problem.
    
    The parameter $\theta$ in \eqref{eq:linear-policy} has to satisfy the conditions described in Theorem~\ref{thm:parameter-condition} so as to satisfy the structural results in Theorem~\ref{thm:monotonicity}. We now define a relation between parameter $\theta\in\R^{\numStops-1}$ in Euclidean coordinates and parameter $\phi\in\R^{\numStops-1}$ in spherical coordinates:
    \begin{align}
    \label{eq:spherical-coordinates}
    \theta_\stopNumber^\phi(i)=\left\{\begin{array}{ll}
    \phi_1^2(1) \prod_{j=l}^{L-1} \sin ^2\left(\phi_{j}(1)\right),& i=1 \\
    1+\phi_1^2(2) \prod_{j=2}^{l} \sin ^2\left(\phi_{j}(2)\right), & i=2 \\
    \theta_l(2) \prod_{j=1}^L \sin ^2\left(\phi_{j}(i)\right),& i>2 .
    \end{array}\right.
    \end{align}
    It can be easily verified that the $\theta$ obtained using~\eqref{eq:spherical-coordinates} satisfies the conditions in Theorem~\ref{thm:parameter-condition}. So, instead of optimizing the parameter $\theta$ using a constrained optimization problem, we can optimize the parameter $\phi$ as an unconstrained optimization problem. 
            
    \subsection{Policy gradient reinforcement learning algorithm \cite{2005:JS}}    
    \label{sec:spsa}
    In this subsection, we describe the policy gradient algorithm to optimize the parameters of the linear mining policy~\eqref{eq:linear-policy} for a blockchain-enabled IoT device. As it is difficult to obtain a closed-form expression for the cumulative reward as a function of the mining policy, we utilize techniques from stochastic optimization to optimize the policy parameters. 
    \begin{algorithm}        
     \caption{Policy gradient algorithm}
     \begin{algorithmic}[1]
     \REQUIRE $\compPowerDistribution_0,\compPowerSet,\obsvSet,\transMatrix,\obsvMatrix,\reward,\discountFactor,\numStops,\epsilon,\zeta,\kappa,\nu,\psi$
      \STATE Initialize $\phi^{(0)}$ randomly.
      \FOR {$n = 1$ to $N$}
      \STATE Compute $\theta^{\phi^{(n)}}$ using~\eqref{eq:spherical-coordinates} and compute $a_n, c_n$ using~\eqref{eq:spsa-parameters}.
      \STATE Use Algorithm~\ref{algo:pomdp-simulator} to simulate the POMDP for the energy-efficient mining problem in Sec.\ref{sec:model} using the linear mining policy~\eqref{eq:linear-policy} with policy parameters $\theta^{\phi^{(n)}}+c_n\omega_n$ and $\theta^{\phi^{(n)}}+c_n\omega_n$. Update the parameter $\phi^{(n)}$ using~\eqref{eq:spsa-algorithm}.
      \ENDFOR
     \RETURN $\theta^{\phi^{(N)}}$ 
     \end{algorithmic} 
     \label{algo:policy-gradient}
     \end{algorithm}
    
    The policy gradient algorithm to optimize the parameter $\theta$ in spherical coordinates is as follows: (1) initialize the parameter $\phi$, (2) update $\phi$ using~\eqref{eq:spsa-algorithm}. Algorithm~\ref{algo:policy-gradient} summarizes the steps in the policy gradient algorithm. 
    \begin{align}
    \label{eq:spsa-algorithm}
    \begin{aligned}
        \hat{\nabla}_{\phi} \totalReward\left(\theta^{\phi^{(n)}}\right)&=\frac{\totalReward\left(\theta^{\phi^{(n)}}+c_n \omega_n\right)-\totalReward\left(\theta^{\phi^{(n)}}-c_n \omega_n\right)}{2 c_n} \omega_n\\
        \phi^{(n+1)}&=\phi^{(n)}+a_n \hat{\nabla}_{\phi}J\left(\theta^{\phi^{(n)}}\right)
    \end{aligned}
    \end{align}
    Here, $\phi^{(n)}$ is the value of parameter $\phi$ at $n^{th}$ iteration. $\theta^{\phi^{(n)}}$ is the value of the parameter in Euclidean coordinates obtained using~\eqref{eq:spherical-coordinates}.
    The parameters $a_n$ and $c_n$ are typically chosen as:
    \begin{align}
    \label{eq:spsa-parameters}
        \begin{array}{llll}
        a_n=\varepsilon(n+1+\varsigma)^{-\kappa}, \quad c_n=\psi(n+1)^{-v},\\ 
        0.5<\kappa \leq 1,\varepsilon,\quad \varsigma>0, 0.5<v \leq 1,\quad \psi>0\\
        \end{array}
    \end{align}
    
    To summarize, we showed that the optimal mining policy has a threshold structure (Theorem~\ref{thm:monotonicity}), and it partitions the belief space into two connected sets. We exploited these results to design a linear mining policy~\eqref{eq:linear-policy} for the energy-efficient mining problem in the blockchain. We specified conditions on the parameters of the linear mining policy (Theorem~\ref{thm:parameter-condition}) so that it satisfies the structural results in Theorem~\ref{thm:monotonicity}. This was followed by the transformation of the parameters of the linear mining policy to spherical coordinates~\eqref{eq:spherical-coordinates}. The latter facilitated us to formulate the optimization of the policy parameters as an unconstrained optimization problem. We also presented a policy gradient algorithm~\eqref{eq:spsa-algorithm} to optimize the linear policy's parameters in the spherical coordinates. 
	
    \section{Numerical results and bitcoin dataset}
    \label{sec:simulations}
    In this section, we compute an optimal linear mining policy~\eqref{eq:linear-policy} for the energy-efficient mining problem in blockchain\footnote{All the numerical results are reproducible, and the codes are available on GitHub at https://github.com/anuraggin/blockchain_pomdp.git}  using synthetic (Sec.\ref{sec:simulation-synthetic-low} and Sec.\ref{sec:simulation-synthetic-high}) and a real bitcoin dataset (Sec.\ref{sec:simulation-real}). The optimal linear mining policy~\eqref{eq:linear-policy} satisfies the structural results in Theorem~\ref{thm:monotonicity} and is suitable for IoT devices: the linear policy uses less memory, less computation and can be computed offline. To illustrate the performance of our proposed optimal linear policy~\ref{pol:linear}, we compare it with four other mining strategies: 
    \begin{enumerate}[leftmargin=*, label=(P\arabic*)]
        \item \label{pol:optimal} \textbf{Optimal mining policy}: a mining policy obtained using the value iteration algorithm for the multiple stopping time problem. This qualifies as the ground truth since it is the optimal solution~\cite{2016:VK}.
        \item \label{pol:linear} \textbf{Optimal linear mining policy}: a linear mining policy~\eqref{eq:linear-policy} obtained using the policy gradient algorithm~\eqref{eq:spsa-algorithm}.
        \item \label{pol:first-L}\textbf{First $\numStops$ mining}: a policy that chooses the first $\numStops$ time instants for mining.
        \item \label{pol:random} \textbf{Random policy}: a policy that decides to mine or not at each time instant by tossing a biased coin. For a random policy, the probability of heads for a biased coin can be adjusted based on the rate of data sensing. In our simulation, we consider a fair coin.
        \item \label{pol:rl}\textbf{Reinforcement learning based mining policy}: Reinforcement learning has been exploited in the literature to compute optimal mining policy under various settings~\cite{2021:TW-et-al}\cite{2023:JS-et-al}; it uses softmax parametrization to model policy. In the reinforcement learning paradigm, the parameters of the MDP are unknown to the decision maker. Hence, the policy is designed as a function of the current and past observations.
            $\operatorname{Pr}(\mu(\pi, l)=u)=\frac{\exp \left( \theta_{l, u}^\intercal W_t\right)}{\sum_{u=1}^2 \exp \left( \theta_{l, u}^\intercal W_t\right)}$
        Here, $W_t :=\left[y_t\quad y_{t-1}\quad\ldots\quad y_{t-N+1}\right]^\intercal$ is the observation window,$\theta_{\stopNumber,\action}\in\R^{N},\forall \stopNumber,\forall\action$ is the policy parameters; $N$ denotes the size of observation window for designing a mining policy using the softmax parametrization. For simulation, we chose $N=2$ so that the number of parameters in the softmax parametrization is similar to that of the linear mining policy.
    \end{enumerate}

    \subsection{Low-dimensional numerical examples using synthetic data}
    \label{sec:simulation-synthetic-low}
    \begin{table}
        \centering
        \caption{Model parameters for the energy-efficient mining problem in blockchain for an IoT device (Low-dimensional synthetic data)}
        \begin{tabular}{c|c|l}
        \hline
            \textbf{Parameters} & \textbf{Eq.} & \textbf{Values}\\
            \hline
            $\{\compPowerDistribution_0,\compPowerSet\}$ & \eqref{eq:trans-matrix} & $\left\{[0\; 0\; 1],\{1,2,3\}\right\}$\\
            $\transMatrix$ & \eqref{eq:trans-matrix} & $\left[\begin{array}{ccc}0.5&0.5&0\\0.25&0.5&0.25\\0&0.5&0.5\end{array}\right]$\\
            $\obsvSet$ & \eqref{eq:obsv-matrix} & $\{1,2,3,4,5\}$\\
            $\obsvMatrix^\intercal$ & \eqref{eq:obsv-matrix} & $\left[\begin{array}
            {ccccc} 0.2384 & 0.1686 & 0.0221\\0.3129 & 0.2580 & 0.0955\\ 0.3951 & 0.3258 & 0.1207\\ 0.0629 & 0.3 & 0.4546\\0.0044 & 0.0669 & 0.1741\end{array}\right]$\\
            $[r(1,2),r(2,2),r(3,2)]$ & \eqref{eq:reward} &[0.1, 0.01, 0.001]\\$[\discountFactor, \numStops]$ & \eqref{eq:cumulative-reward} & $[0.9,3]$\\$[\epsilon, \zeta, \kappa, \nu, \psi]$ & \eqref{eq:spsa-parameters} & $[0.7,0.1,0.6,0.6,0.1]$\\
            \hline
        \end{tabular}        
        \label{tab:model-parameters-simple}
    \end{table}
    
    \begin{figure}
        \centering
        \begin{subfigure}{0.45\textwidth}
        \includegraphics[width=0.9\textwidth]{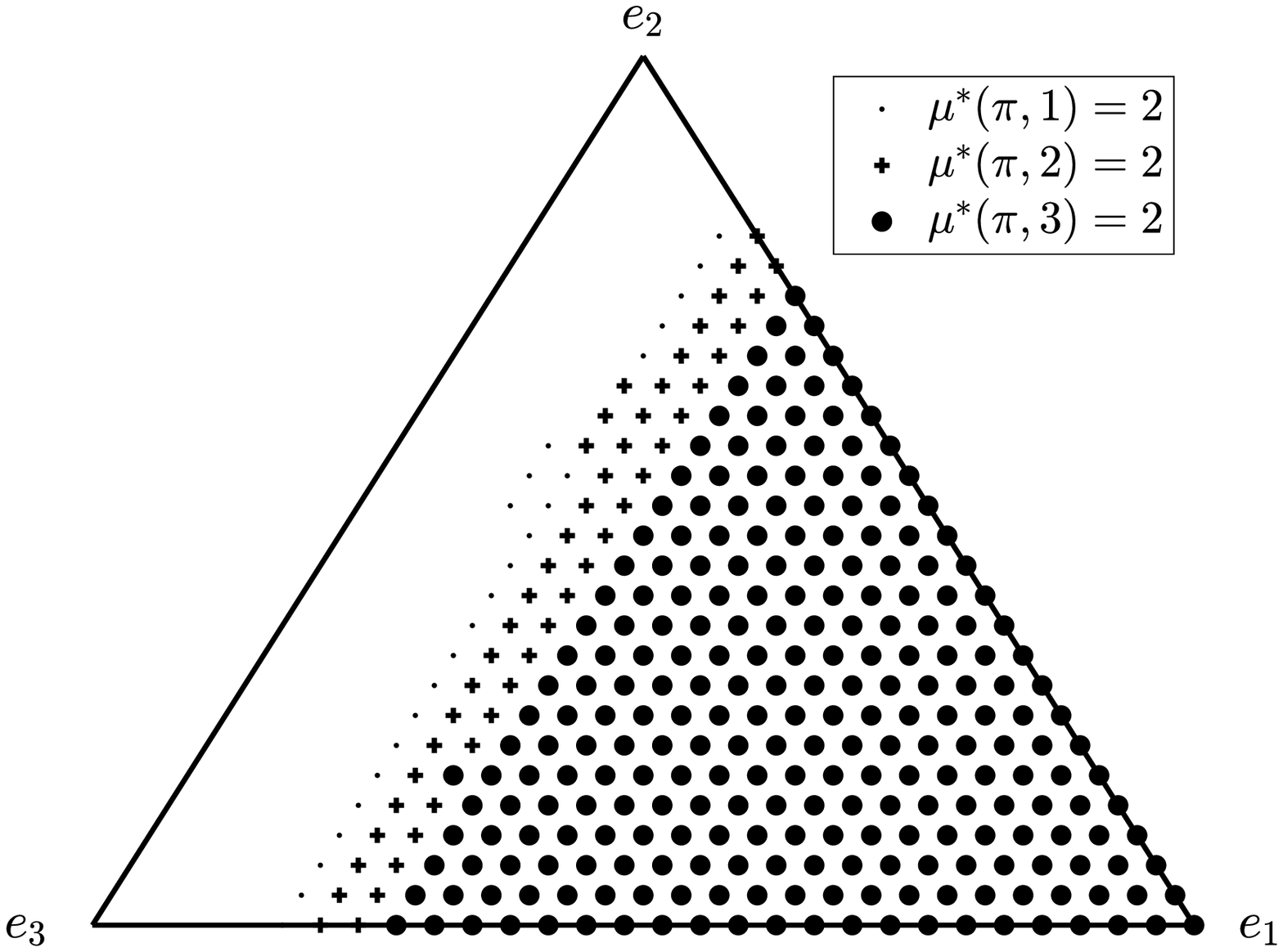}
        \caption{Value iteration algorithm.}
        \label{fig:value-iteration}
        \end{subfigure}
        \begin{subfigure}{0.45\textwidth}
            \includegraphics[width=0.9\textwidth]{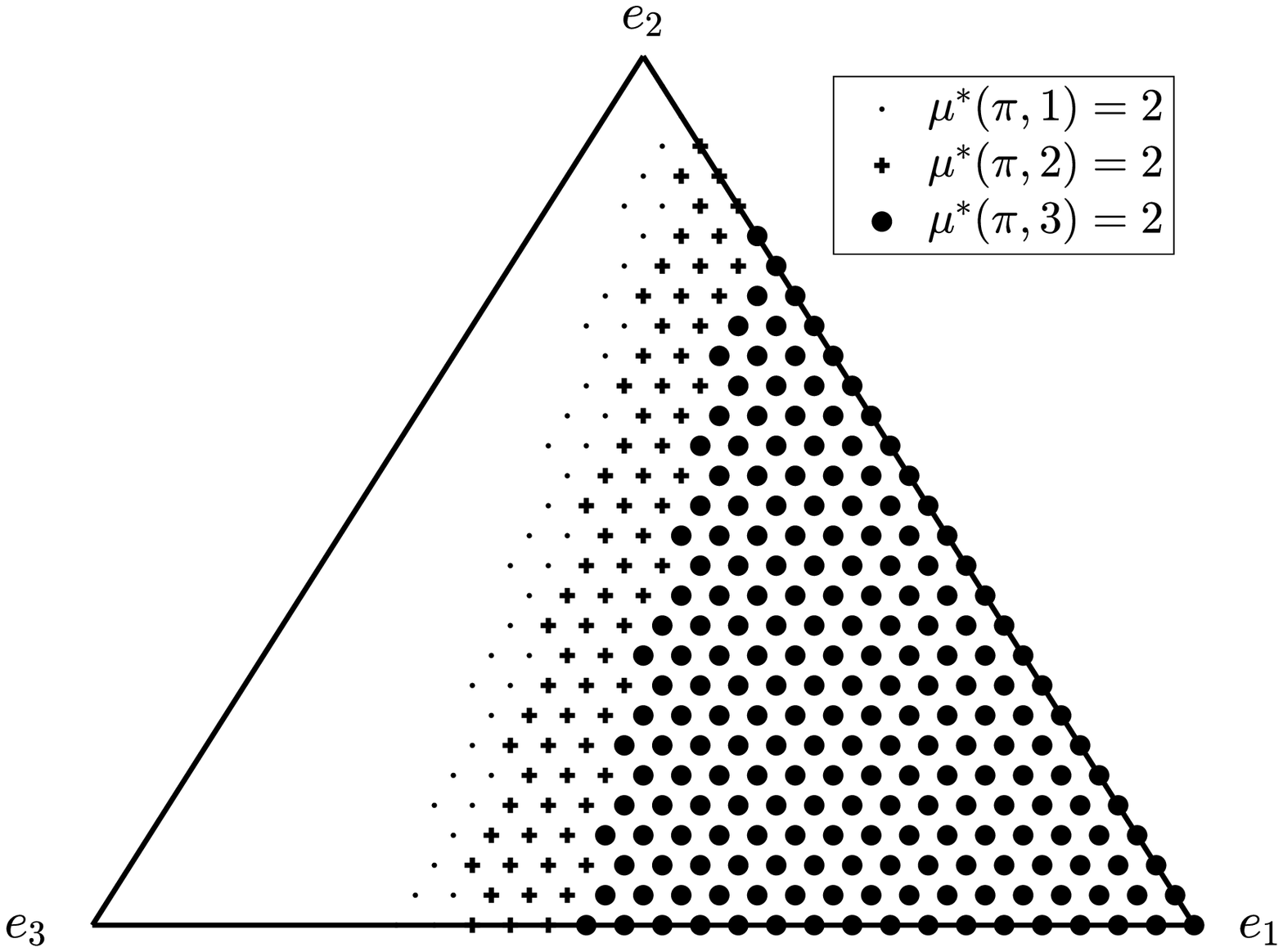}
            \caption{Policy gradient algorithm.}
            \label{fig:linear-policy}
        \end{subfigure}
        \caption{Optimal mining policy for the energy-efficient mining problem in blockchain on a low-dimensional synthetic data. The triangle represents the belief space for the energy-efficient mining problem in blockchain. The belief space has been discretized into 30 equal parts along each axis. The markers indicate the belief states where the optimal action is to mine in the blockchain. The optimal mining policy has a threshold and nested structure (Theorem~\ref{thm:monotonicity}).}
    \end{figure}
    
    \begin{figure*}[t]
        \centering
        \begin{subfigure}{0.45\textwidth}
        \includegraphics[width=0.9\textwidth]{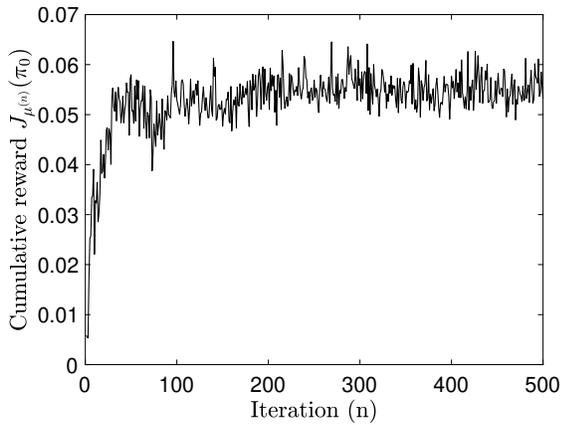}
        \caption{Low-dimensional synthetic data (Sec.\ref{sec:simulation-synthetic-low})}
        \label{fig:spsa-convergence}
        \end{subfigure}
        \begin{subfigure}{0.45\textwidth}
            \includegraphics[width=0.9\textwidth]{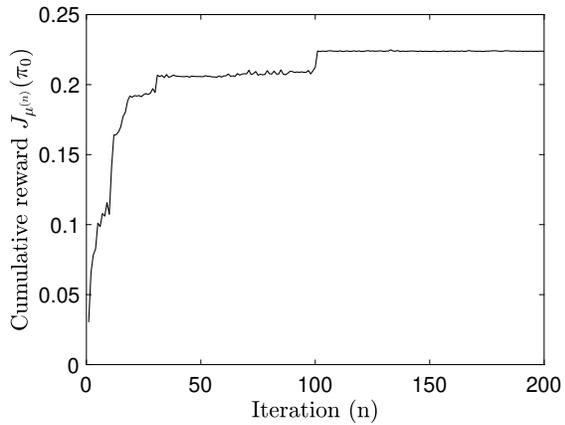}
            \caption{High-dimensional synthetic data (Sec.\ref{sec:simulation-synthetic-high}).}
            \label{fig:spsa-convergence-large}
        \end{subfigure}
        \caption{Cumulative reward for the energy-efficient mining problem at each iteration of the policy gradient algorithm. $\policy^{(n)}$ denotes the linear mining policy $\policy$ at iteration $n$. The policy gradient algorithm converges within 100 iterations. Therefore, if the total miners in an IoT application evolve slowly with time, the IoT device can also use the policy gradient algorithm to update its optimal mining policy.}
    \end{figure*}

    \begin{table}
    \centering
    \caption{Comparison of the optimal mining policy with other heuristic mining strategies}
    \begin{subtable}{0.45\textwidth}
        \centering
        \caption{Low-dimensional synthetic dataset}
        \begin{tabular}{l|c}
            \hline
            \textbf{Policy} & \textbf{Reward~\eqref{eq:cumulative-reward}}\\
            \hline
            \ref{pol:optimal} Optimal mining policy &  0.0579\\
            \ref{pol:linear} Optimal linear mining policy &  0.0549\\
            \ref{pol:first-L} First $\numStops$ mining & 0.0204\\
            \ref{pol:random} Random policy & 0.0348\\
            \ref{pol:rl} Reinforcement learning & 0.0452\\
            \hline
        \end{tabular}
        \label{tab:different-policy}
    \end{subtable}\\\vspace{0.2cm}
    \begin{subtable}{0.45\textwidth}
        \centering
        \caption{Real Bitcoin dataset}
        \begin{tabular}{l|c}
            \hline
            \textbf{Policy} & \textbf{Reward~\eqref{eq:cumulative-reward}}\\
            \hline
            \ref{pol:optimal} Optimal mining policy &  0.2021\\
            \ref{pol:linear} Optimal linear mining policy &  0.1991\\
            \ref{pol:first-L} First $\numStops$ mining & 0.1265\\
            \ref{pol:random} Random policy & 0.1432\\
            \ref{pol:rl} Reinforcement learning & 0.1621\\
            \hline
        \end{tabular}
        \label{tab:different-policy-real}
    \end{subtable}
    \end{table}
    
    In this subsection, we use synthetic model parameters for the proposed energy-efficient mining problem in blockchain (Sec.\ref{sec:model}). Our model parameters are summarized in Table~\ref{tab:model-parameters-simple}. We chose $\compPowerSetSize=3$ to visualize the structure of the optimal mining policy. We solved the optimal mining policy for the energy-efficient mining problem in blockchain using the value iteration algorithm (Fig.~\ref{fig:value-iteration}) and the optimal linear mining policy using the policy gradient algorithm (Fig.~\ref{fig:linear-policy}). Fig.~\ref{fig:spsa-convergence} shows the convergence of the policy gradient algorithm.
    The optimal mining policy gives an expected reward of 0.0579 whereas the optimal linear mining policy gives an expected reward of 0.0549. Therefore, there is a 5.5\% loss in the expected reward for using a linear mining policy. On the positive side, the optimal linear mining policy uses much less memory and requires less computation compared to the solution of the value iteration algorithm. This is because the solution of value iteration corresponds to storing a look-up table, the size of which grows exponentially with the size of state space. Moreover, for the optimal mining policy, the IoT device has to perform a search operation on the look-up table at each time instant to obtain the optimal action. The optimal linear mining policy overcomes these two drawbacks making it suitable for resource-constrained IoT devices.
    
    We also compared our proposed mining policy with other heuristic mining strategies. The results are summarized in Table~\ref{tab:different-policy}. We observe that the optimal linear mining policy provides a significant improvement over other heuristic policies: $69\%$ improvement over the first $\numStops$ mining policy, $58\%$ improvement over the random policy, and $22\%$ improvement over reinforcement learning-based mining policy.

    \subsection{High-dimensional numerical example using synthetic data}
    \label{sec:simulation-synthetic-high}
    In this subsection, we solve a higher dimensional energy-efficient mining problem in blockchain (Sec.\ref{sec:model}) using synthetic data. The model parameters are summarized in Table~\ref{tab:model-parameters-large}.
    Fig.\ref{fig:spsa-convergence-large} shows the convergence of the policy gradient algorithm. Even for high dimensional data, the policy gradient algorithm converges within 200 iterations using a suitable choice of parameters for the policy gradient algorithm. Therefore, if the total number of miners in an IoT application evolves slowly with time, the IoT device can also use the policy gradient algorithm to update its optimal linear mining policy.
    \begin{table}
        \centering
        \caption{Model parameters for the energy-efficient mining problem in blockchain for an IoT device (High-dimensional synthetic data)}
        \begin{tabular}{c|c|l}
        \hline
            \textbf{Parameters} & \textbf{Eq.} & \textbf{Values}\\
            \hline
            $\{\compPowerDistribution_0,\compPowerSet\}$ & \eqref{eq:trans-matrix} & $\left\{[0\; 0 \;\ldots \;0\; 1],\{1,2,\ldots,10\}\right\}$\\
            $\transMatrix$ & \eqref{eq:trans-matrix} & $\transMatrix_{i,i}= 0.5, \;\forall i,\;\transMatrix_{1,2}=\transMatrix_{10,9}=0.5$,\\&&$\transMatrix_{i,j-1}=\transMatrix_{i,j+1}=0.25,\forall i\neq\{1,10\}$\\
            $\{\obsvSet,\obsvMatrix(i,\obsv)\}$ & \eqref{eq:obsv-matrix} & $\left\{\{1,2,\ldots,12\},\frac{(10\,i)^\obsv\exp(-10\,i)/\obsv!}{\sum_\obsv(10i)^\obsv\exp(-10i)/\obsv!}\right\}$\\
            $\reward(\compPowerRV,2)$ & \eqref{eq:reward} & $1/\compPowerRV^3$\\$[\discountFactor, \numStops]$ & \eqref{eq:cumulative-reward} & $[0.9,3]$\\$[\epsilon, \zeta, \kappa, \nu, \psi]$ & \eqref{eq:spsa-parameters} & $[0.7,0.1,0.6,0.6,0.1]$\\
            \hline
        \end{tabular}        
        \label{tab:model-parameters-large}
    \end{table}
    \subsection{Numerical examples using real bitcoin dataset}
    \label{sec:simulation-real}
    \begin{figure*}[t]        
        \centering
        \begin{subfigure}{0.45\textwidth}
        \includegraphics[width=0.9\textwidth]{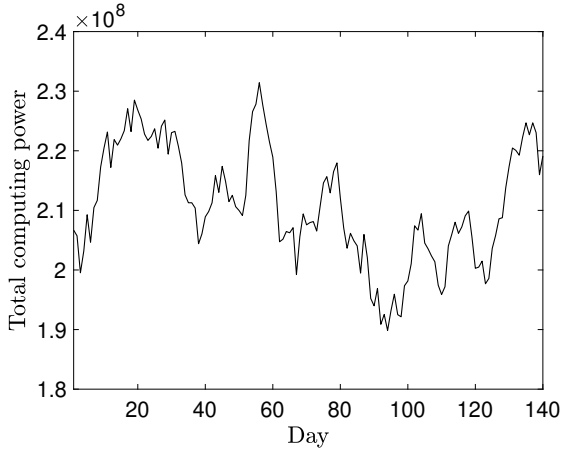}
        \caption{Total computing power vs. time. The total computing power was estimated using the rate of new blocks in the blockchain over 24 hours interval.}
        \label{fig:data_total_computing}
        \end{subfigure}\hspace{0.5cm}
        \begin{subfigure}{0.45\textwidth}
        \includegraphics[width=0.9\textwidth]{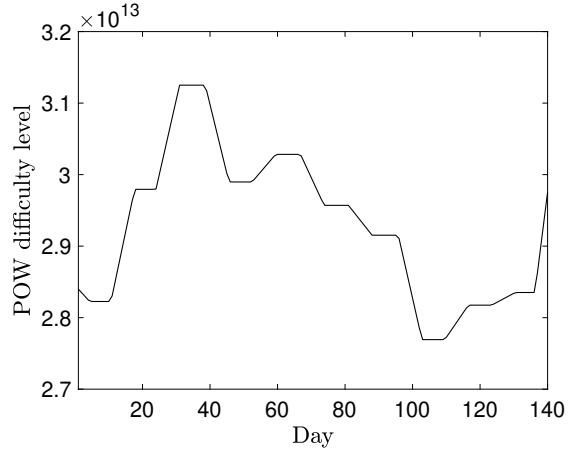}
        \caption{PoW difficulty level vs. time. The PoW difficulty level is the average of the PoW difficulty level computed over 24-hour interval.}
        \label{fig:data_pow}
        \end{subfigure}
        \caption{Bitcoin mining dataset between April 2022 - August 2022 (Source: \cite{2023:blockchain})}.
        \label{fig:data}
    \end{figure*}
    Now, we use real Bitcoin data to estimate the model parameters for the energy-efficient mining problem in blockchain (Sec.\ref{sec:model}). The estimated model parameters are used to solve the optimal mining policy for a blockchain-enabled IoT device. 
    
    A record of historical data on bitcoin mining is available in~\cite{2023:blockchain}. We use their data on the estimated hash rate (total computing power) and the difficulty (PoW difficulty level) to estimate the transition matrix and the observation matrix. The dataset contains total computing power and the PoW difficulty sampled once per day. As the total number of miners in the Bitcoin network is growing with time, we use a small range of data to compute the model parameters.
    Fig.~\ref{fig:data} shows the plot of the Bitcoin dataset between April 2022 - August 2022.
    To compute the model parameters, we first group the data into bins of uniform size as follows: (1) total computing power is grouped into three bins, i.e., $\compPowerSetSize=3$, (2) PoW difficulty level is grouped into five bins, i.e., $\obsvSetSize=5$. We used the binned data to estimate the transition matrix $\transMatrix$ 
    using the MLE estimator (see Table~\ref{tab:model-parameters-real}).
    The estimated transition matrix $\transMatrix$ satisfies the tri-diagonal structure assumption, and the diagonal terms are dominant, thereby satisfying the assumption \ref{assumption:slow-mc}. Hence, \ref{assumption:slow-mc} is easy to satisfy with a suitable choice of binning and sampling interval. However, in our case, the estimated observation matrix using the MLE did not yield a TP2 matrix. This could be due to external factors or the insufficient size of the dataset.
    Hence, to exploit our structural results, we estimate the observation matrix $\hat{\obsvMatrix}$ within the class of TP2 distribution\footnote{Although it is beyond the scope of this study, it would be worth studying the loss in optimality due to estimation of the model parameters within the class of TP2 distribution.}. The model parameters are summarized in Table~\ref{tab:model-parameters-real}.
    \begin{table}
        \centering
        \caption{Model parameters for the energy-efficient mining problem in blockchain for an IoT device (Real Bitcoin dataset)}
        \begin{tabular}{c|c|l}
        \hline
            \textbf{Parameters} & \textbf{Eq.} & \textbf{Values}\\
            \hline
            $\{\compPowerDistribution_0, \compPowerSet\}$, & \eqref{eq:trans-matrix} & $\left\{[0\; 0\; 1],\{1,2,3\}\right\}$\\
            $\transMatrix$ & \eqref{eq:trans-matrix} & $\left[\begin{array}{ccc}0.8&0.2&0\\0.038&0.8861&0.0759\\0&0.1111&0.8889\end{array}\right]$\\
            $\obsvSet$ & \eqref{eq:obsv-matrix} & $\{1,2,3,4,5\}$\\
            $\obsvMatrix^\intercal$ & \eqref{eq:obsv-matrix} & $\left[\begin{array}{ccc}0.2384&0.1686&0.0221\\0.3129&0.258&0.0955\\0.3951&0.3258&0.1207\\0.0629&0.3&0.4546\\0.0044&0.0669&0.1741\end{array}\right]$\\
            $[r(1,2),r(2,2),r(3,2)]$ & \eqref{eq:reward} &[1, 0.125, 0.037]\\$[\discountFactor, \numStops]$ & \eqref{eq:cumulative-reward} & [0.9, 3]\\$[\epsilon, \zeta, \kappa, \nu, \psi]$ & \eqref{eq:spsa-parameters} & [0.5, 0.1, 0.6, 0.6, 0.1]\\
            \hline
        \end{tabular}        
        \label{tab:model-parameters-real}
    \end{table}
    
    \begin{figure}
        \centering
        \begin{subfigure}{0.45\textwidth}
        \includegraphics[width=0.9\textwidth]{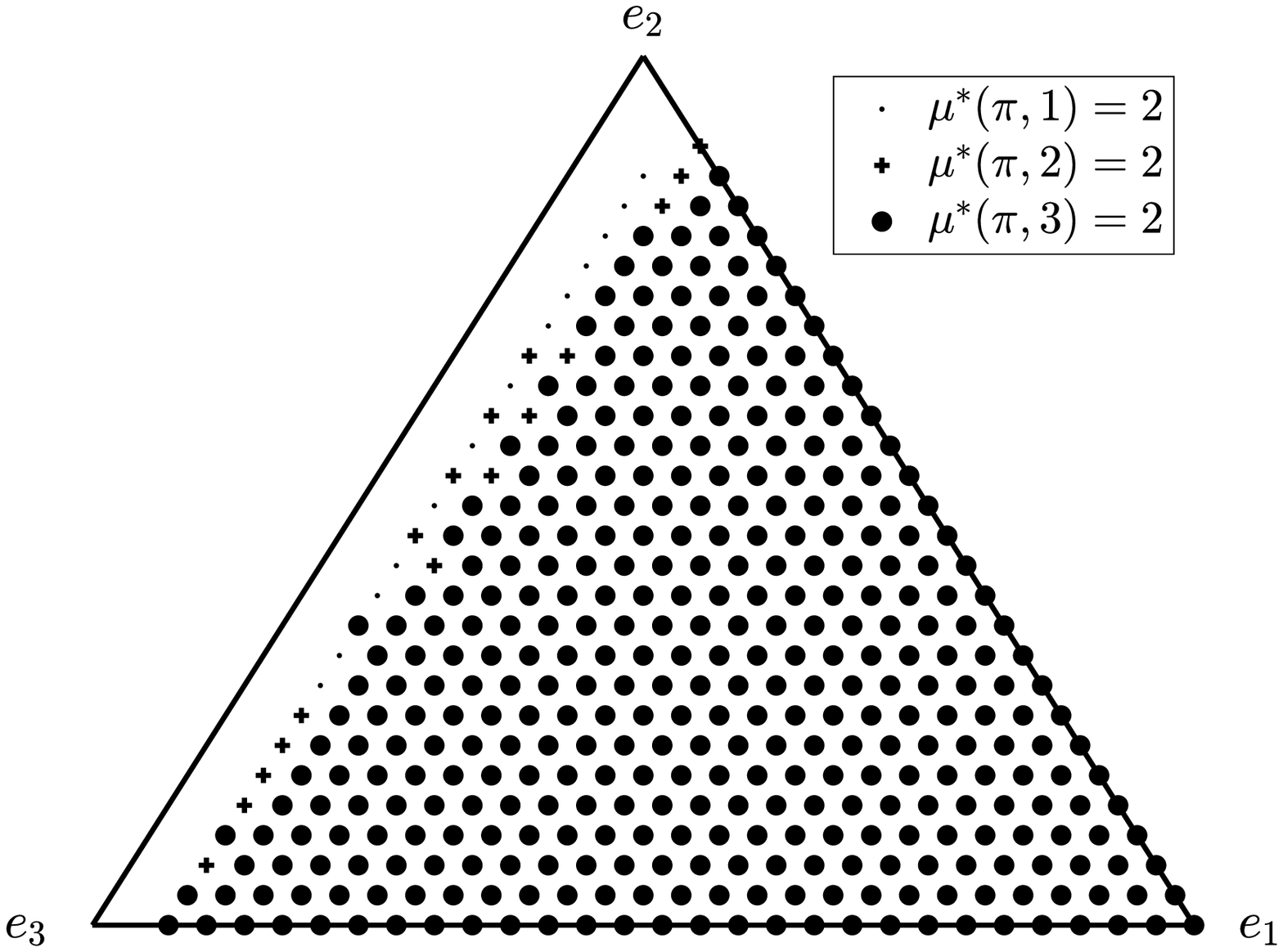}
        \caption{Value iteration algorithm.}
        \label{fig:value-iteration-real}
        \end{subfigure}
        \begin{subfigure}{0.45\textwidth}
            \includegraphics[width=0.9\textwidth]{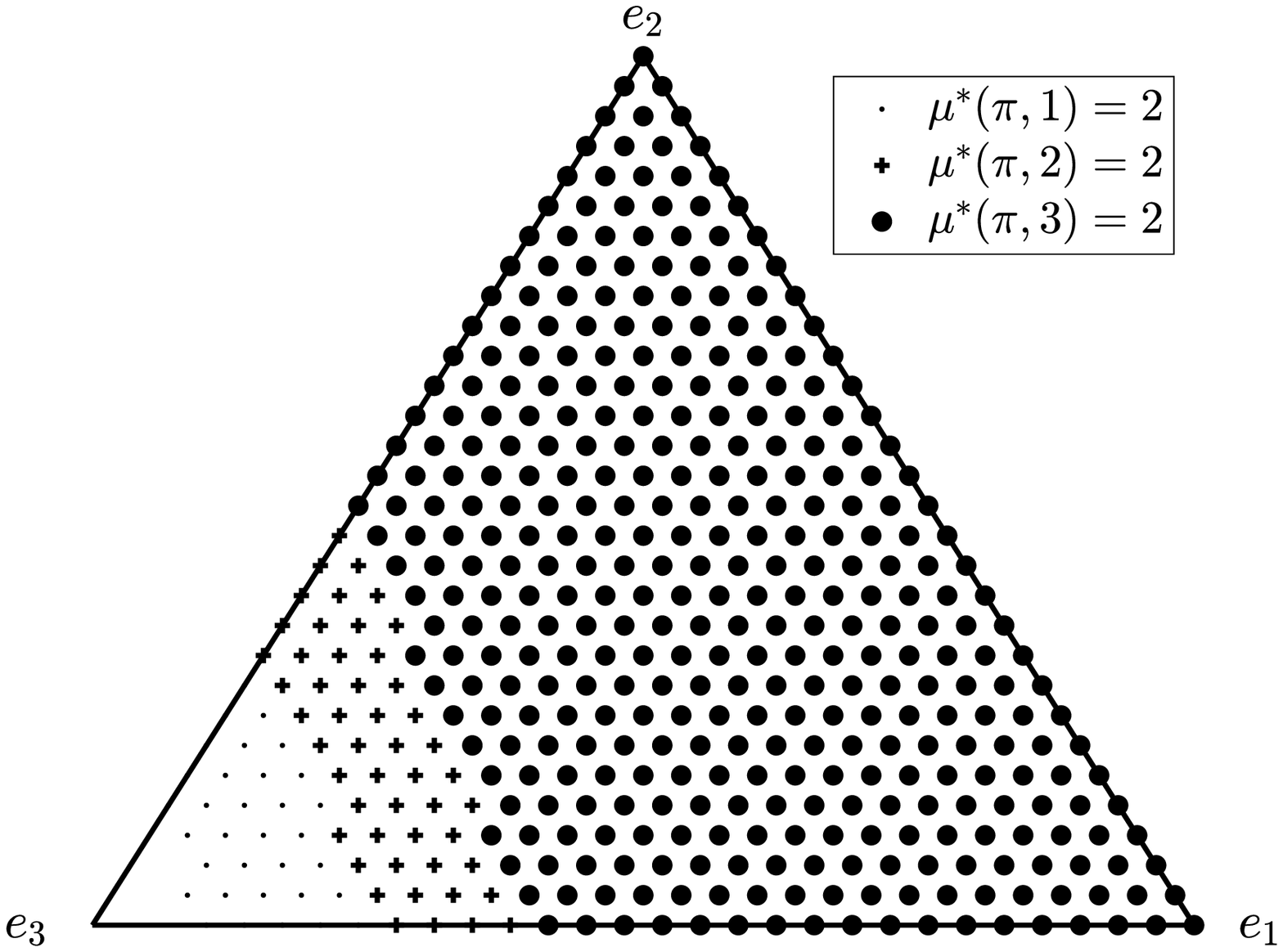}
            \caption{Policy gradient algorithm.}
            \label{fig:linear-policy-real}
        \end{subfigure}
        \caption{Optimal mining policy for energy-efficient mining problem in blockchain on a real Bitcoin dataset. The triangle represents the belief space for the energy-efficient mining problem in blockchain. The belief space has been discretized into 30 equal parts along each axis. The markers indicate the belief states where the optimal action is to mine in the blockchain. The optimal mining policy has a threshold and nested structure (Theorem~\ref{thm:monotonicity}).}
    \end{figure}

    We solved the optimal mining policy for the energy-efficient mining problem in blockchain using the value iteration algorithm (Fig.~\ref{fig:value-iteration-real}) and the optimal linear policy using the policy gradient algorithm (Fig.~\ref{fig:linear-policy-real}). One can observe that the optimal linear mining policy provides a similar performance as that of the optimal mining policy. Furthermore, the linear policy is suitable for resource-constrained IoT devices: it uses less memory and less computation to compute the optimal action.
    
    We also compared the optimal linear mining policy with other heuristic policies on the real Bitcoin dataset. The results are summarized in Table~\ref{tab:different-policy-real}. We observe that the optimal linear mining policy provides a significant improvement over other heuristic policies: $57\%$ improvement over the first $\numStops$ mining policy, $39\%$ improvement over the random policy, and $22\%$ improvement over reinforcement learning-based mining policy.

    To sum up, we solved the optimal mining policy for an energy-efficient mining problem in a blockchain device using real and synthetic model parameters. 
    Additionally, we conducted numerical experiments to compare the optimal mining policy with other heuristic policies.

    \section{Conclusion}
    \label{sec:conclusion}
    In this paper, we addressed: what are the optimal times for an IoT device to mine in a blockchain? We formulated the energy-efficient mining problem in blockchain as a multiple stopping time Bayesian sequential detection problem in POMDP. Computing the exact solution of the multiple stopping time POMDP is computationally intractable. So, using submodularity, we derived a useful mathematical structure that characterizes the optimal mining policy: the optimal policy is monotone in the belief state with respect to MLR order and, therefore, has a threshold structure. We exploited these structural results to derive necessary and sufficient conditions on the parameters of an optimal linear mining policy for the energy-efficient mining problem. This optimal linear policy can be computed offline and stored on the IoT device. This makes it suitable for IoT applications. Finally, we illustrated how the proposed multiple stopping time approach achieves energy-efficient mining in blockchain on synthetic data and a real Bitcoin dataset. We also studied the benefit of the optimal mining policy over other heuristic strategies for a blockchain-enabled IoT device. 
    
    \textbf{Acknowledgement}
    This research was supported in part by the U.S.\ Army Research Office grant  
    W911NF-21-1-0093 and National Science Foundation grant CCF-2112457.    
    \bibliography{anurag}{}

\begin{thebibliography}{10}

\bibitem{2020:AS}
A.~Sunyaev.
\newblock Distributed ledger technology.
\newblock {\em Internet computing: Principles of distributed systems and
  emerging internet-based technologies}, pages 265--299, 2020.

\bibitem{2008:SN}
S.~Nakamoto.
\newblock {Bitcoin}: A peer-to-peer electronic cash system.
\newblock {\em Decentralized business review}, page 21260, 2008.

\bibitem{2021:SM-et-al}
S.~Misra, A.~Mukherjee, A.~Roy, N.~Saurabh, Y.~Rahulamathavan, and
  M.~Rajarajan.
\newblock Blockchain at the edge: Performance of resource-constrained {IoT}
  networks.
\newblock {\em IEEE Transactions on Parallel and Distributed Systems},
  32(1):174--183, 2021.

\bibitem{2018:ON}
O.~Novo.
\newblock Blockchain meets {IoT}: An architecture for scalable access
  management in {IoT}.
\newblock {\em IEEE internet of things journal}, 5(2):1184--1195, 2018.

\bibitem{2023:JB}
J.~Bradshaw.
\newblock Using blockchain technology to increase transparency in agriculture,
  Jan 2023.

\bibitem{2023:EK}
E.~Kovalenko.
\newblock How can the blockchain secure {IoT} networks?, Feb 2023.

\bibitem{2021:LB-et-al}
L.~Badea and M.~C. Mungiu-Pupazan.
\newblock The economic and environmental impact of {Bitcoin}.
\newblock {\em IEEE Access}, 9:48091--48104, 2021.

\bibitem{2016:VK}
V.~Krishnamurthy.
\newblock {\em Partially Observed {Markov} Decision Processes: From Filtering
  to Controlled Sensing}.
\newblock Cambridge University Press, 2016.

\bibitem{1996:DB-JT}
J.~N.~T. Dimitri P.~Bertsekas.
\newblock {\em Neuro-Dynamic Programming}.
\newblock Optimization and neural computation series. Athena Scientific, 1
  edition, 1996.

\bibitem{2022:AA-et-al}
A.~Z. Abyaneh, N.~Zorba, and B.~Hamdaoui.
\newblock Empowering next-generation {IoT} wlans through blockchain and
  802.11ax technologies.
\newblock {\em IEEE Transactions on Intelligent Transportation Systems}, pages
  1--10, 2022.

\bibitem{2020:NK-et-al}
N.~Kullig, P.~Lämmel, and N.~Tcholtchev.
\newblock Prototype implementation and evaluation of a blockchain component on
  {IoT} devices.
\newblock {\em Procedia Computer Science}, 175:379--386, 2020.
\newblock The 17th International Conference on Mobile Systems and Pervasive
  Computing (MobiSPC),The 15th International Conference on Future Networks and
  Communications (FNC),The 10th International Conference on Sustainable Energy
  Information Technology.

\bibitem{2022:XS-et-al}
X.-S. Song, Q.-L. Li, Y.-X. Chang, and C.~Zhan.
\newblock A {Markov} process theory for network growth processes of {DAG}-based
  blockchain systems.
\newblock {\em arXiv preprint arXiv:2209.01458}, 2022.

\bibitem{2020:YL-et-al}
Y.~Liu, S.~Zhang, X.~Chen, X.~Zhou, and X.~Zheng.
\newblock {\em Blockchain Security Analysis: A POMDP-Based Approach for
  Analyzing Blockchain System Security Against the Long Delay Attack}.
\newblock Eliva Press, 2020.

\bibitem{2022:KK-et-al}
K.~Kim, S.-Y.~T. Lee, and S.~Assar.
\newblock The dynamics of cryptocurrency market behavior: sentiment analysis
  using {Markov} chains.
\newblock {\em Industrial Management \& Data Systems}, 122(2):365--395, 2022.

\bibitem{2020:RS-et-al}
R.~Singh, A.~D. Dwivedi, G.~Srivastava, A.~Wiszniewska-Matyszkiel, and
  X.~Cheng.
\newblock A game theoretic analysis of resource mining in blockchain.
\newblock {\em Cluster Computing}, 23:2035--2046, 2020.

\bibitem{2020:GY-et-al}
G.~Yang, Y.~Wang, Z.~Wang, Y.~Tian, X.~Yu, and S.~Li.
\newblock Ipbsm: An optimal bribery selfish mining in the presence of
  intelligent and pure attackers.
\newblock {\em International Journal of Intelligent Systems},
  35(11):1735--1748, 2020.

\bibitem{2021:TW-et-al}
T.~Wang, S.~C. Liew, and S.~Zhang.
\newblock When blockchain meets {AI}: Optimal mining strategy achieved by
  machine learning.
\newblock {\em International Journal of Intelligent Systems}, 36(5):2183--2207,
  2021.

\bibitem{2016:AK-et-al}
A.~Kiayias, E.~Koutsoupias, M.~Kyropoulou, and Y.~Tselekounis.
\newblock Blockchain mining games.
\newblock In {\em Proceedings of the 2016 ACM Conference on Economics and
  Computation}, pages 365--382, 2016.

\bibitem{2023:YZ-et-al}
Y.~Zhang, M.~Liu, J.~Guo, Z.~Wang, Y.~Wang, T.~Liang, and S.~K. Singh.
\newblock Optimal revenue analysis of the stubborn mining based on {Markov}
  decision process.
\newblock In {\em International Conference on Machine Learning for Cyber
  Security}, pages 299--308. Springer, 2023.

\bibitem{2011:MR}
M.~Rosenfeld.
\newblock Analysis of {Bitcoin} pooled mining reward systems.
\newblock {\em arXiv preprint arXiv:1112.4980}, 2011.

\bibitem{2021:WM-et-al}
W.~Mao, Z.~Zhao, Z.~Chang, G.~Min, and W.~Gao.
\newblock Energy-efficient industrial internet of things: Overview and open
  issues.
\newblock {\em IEEE Transactions on Industrial Informatics}, 17(11):7225--7237,
  2021.

\bibitem{2019:CS-et-al}
C.~Savaglio, P.~Gerace, G.~Di~Fatta, and G.~Fortino.
\newblock Data mining at the {IoT} edge.
\newblock In {\em 2019 28th International Conference on Computer Communication
  and Networks (ICCCN)}, pages 1--6, 2019.

\bibitem{2018:VK-et-al}
V.~Krishnamurthy, A.~Aprem, and S.~Bhatt.
\newblock Multiple stopping time {POMDP}s: Structural results \& application in
  interactive advertising on social media.
\newblock {\em Automatica}, 95:385--398, 2018.

\bibitem{2023:TB-UM}
T.~Bozkus and U.~Mitra.
\newblock Link analysis for solving multiple-access {MDPs} with large state
  spaces.
\newblock {\em IEEE Transactions on Signal Processing}, 71:947--962, 2023.

\bibitem{2014:VK-RJ}
V.~Krishnamurthy and C.~Rojas.
\newblock Reduced complexity {HMM} filtering with stochastic dominance bounds:
  A convex optimization approach.
\newblock {\em IEEE Transactions on Signal Processing}, 62(23):6309--6322,
  2014.

\bibitem{2022:KH-RS}
K.~Hammar and R.~Stadler.
\newblock Intrusion prevention through optimal stopping.
\newblock {\em IEEE Transactions on Network and Service Management},
  19(3):2333--2348, 2022.

\bibitem{2017:DS-UM}
D.-S. Zois and U.~Mitra.
\newblock Active state tracking with sensing costs: Analysis of two-states and
  methods for $n$-states.
\newblock {\em IEEE Transactions on Signal Processing}, 65(11):2828--2843,
  2017.

\bibitem{2011:GA-et-al}
G.~Atia, V.~Veeravalli, and J.~Fuemmeler.
\newblock Sensor scheduling for energy-efficient target tracking in sensor
  networks.
\newblock {\em IEEE Transactions on Signal Processing}, 59(10):4923--4937,
  2011.

\bibitem{2017:GR-et-al}
G.~Rovatsos, X.~Jiang, A.~D. Domínguez-García, and V.~V. Veeravalli.
\newblock Statistical power system line outage detection under transient
  dynamics.
\newblock {\em IEEE Transactions on Signal Processing}, 65(11):2787--2797,
  2017.

\bibitem{2014:NH}
N.~Houy.
\newblock The bitcoin mining game.
\newblock {\em Available at SSRN 2407834}, 2014.

\bibitem{2017:DN}
N.~Dimitri.
\newblock Bitcoin mining as a contest.
\newblock {\em Ledger}, 2:31--37, 2017.

\bibitem{2021:NU-et-al}
N.~A. Urpi, S.~Curi, and A.~Krause.
\newblock Risk-averse offline reinforcement learning.
\newblock {\em arXiv preprint arXiv:2102.05371}, 2021.

\bibitem{2018:DF-KM}
D.~Fakhri and K.~Mutijarsa.
\newblock Secure {IoT} communication using blockchain technology.
\newblock In {\em 2018 international symposium on electronics and smart devices
  (ISESD)}, pages 1--6. IEEE, 2018.

\bibitem{2017:CD}
C.~Dannen.
\newblock {\em Introducing Ethereum and solidity}, volume~1.
\newblock Springer, 2017.

\bibitem{2023:blockchain}
Blockchain charts, 2023.

\bibitem{2002:AM-DS}
A.~Muller and D.~Stoyan.
\newblock {\em Comparison Methods for Stochastic Models and Risk}.
\newblock Wiley, 2002.

\bibitem{2015:WA}
W.~Anwaar and M.~A. Shah.
\newblock Energy efficient computing: A comparison of {Raspberry} pi with
  modern devices.
\newblock {\em Energy}, 4(02), 2015.

\bibitem{2005:JS}
J.~C. Spall.
\newblock {\em Introduction to stochastic search and optimization: estimation,
  simulation, and control}.
\newblock John Wiley \& Sons, 2005.

\bibitem{2023:JS-et-al}
J.~Soria, J.~Moya, and A.~Mohazab.
\newblock Optimal mining in proof-of-work blockchain protocols.
\newblock {\em Finance Research Letters}, 53:103610, 2023.

\bibitem{2011:DT}
D.~M. Topkis.
\newblock {\em Supermodularity and Complementarity}.
\newblock Princeton University Press, 2011.

\end{thebibliography}
    \bibliographystyle{unsrtAbbv}

    \clearpage
    \section*{Supplementary material}
    \section{Proof of Theorem~\ref{thm:monotonicity} in Sec.\ref{sec:structural-results}}
    \label{proof:thm-monotonicity}
    
The value iteration algorithm is an iterative approach to solve Bellman's equation. However, in this paper, we use the value iteration algorithm to prove by  mathematical induction that the optimal policy has a threshold structure. The value iteration algorithm is as follows:
    \begin{align*}
        V_{k+1}(\pi, l)&=\max _{u \in\{1,2\}} Q_{k+1}(\pi, l, u)\\
        \mu_{k+1}(\pi, l)&=\arg \max_{u \in\{1,2\}}  Q_{k+1}(\pi, l, u), \\
        Q_{k+1}(\pi,\stopNumber,2) &= r^\intercal\pi+\rho\sum_y V_k(T(\pi,y),l+1)\sigma(\pi,y)\\
        Q_{k+1}(\pi,\stopNumber,1) &= \rho\sum_y V_k(T(\pi,y),l)\sigma(\pi,y)
    \end{align*}
        Here, $r:=[r(x_1,2)\; r(x_2,2)\; \ldots\; r(x_{\compPowerSetSize},2)]$ is the reward vector when the IoT device decides to mine.
    Define $W_k(\pi,l):=V_k(\pi,l)-V_k(\pi,l+1)$. The mine set $M^l_k$ and don't mine set $D^l_k$ to decide $l^{th}$ mining time instant at iteration $k$ of the value iteration algorithm is defined as:
    \begin{align}
    \label{eq:set-definition}
        \begin{aligned}            
            M^l_{k+1}&=\{\pi|r^\intercal\pi\geq\rho\sum_y W_k(T(\pi,y),l)\sigma(\pi,y)\}\\
            D^l_{k+1}&=\{\pi|r^\intercal\pi<\rho\sum_y W_k(T(\pi,y),l)\sigma(\pi,y)\}        
        \end{aligned}
    \end{align}
    Our main result uses the following Lemmas from~\cite{2018:VK-et-al}.
    \begin{lemma}
    \label{lemma:w-increasing}
    $W_k(\pi,l)\leq W_k(\pi,l+1)$, and
    $M^l_{k}\supset M^{l+1}_{k}$
    \end{lemma}
	\begin{proof}  The proof is by forward-induction over $k$ and backward-induction over $\stopNumber$, i.e., we assume that the lemma holds for $\stopNumber+1,\ldots,\numStops$ for all values of $k$, and upto $k$ for $\stopNumber$. The base case for mathematical induction is $W_0(\cdot,\cdot)=0,\,\stopset_k^{L+1}=\phi,\,\forall k$.
		The induction step is as follows. Note that 
		\begin{align*}
		W_{k}(\pi, l) & =\bigg[\rho \sum_y W_{k-1}(T(\pi, y), l) \sigma(\pi, y)\bigg] \,\mathcal{I}_{\contset_{k}^l}(\pi) \\
		&+\reward ^\intercal \compPowerDistribution\, \mathcal{I}_{\contset_{k}^{l+1} \,\cap \,\stopset_{k}^l}(\pi) \\
		& +\bigg[\rho \sum_y W_{k-1}(T(\pi, y), l+1) \sigma(\pi, y)
		\bigg] \mathcal{I}_{\stopset_{k}^{l+1}}(\pi) 
		\end{align*}    
				Consider the following cases:\\
		(a) $\pi \in \stopset_{k}^{l+2}$. This implies $\pi \in \stopset_{k}^{l+1}, \stopset_{k}^{l}$ and $\pi \notin \contset_{k}^{l+2}, \contset_{k}^{l+1},\contset_{k}^{l}$.
		\begin{align*}
		W_{k+1}(\pi,l) - W_{k+1}(\pi,l+1) = \\
		\rho \sum_\obsv \big( W_{k}(T(\pi,y),l+1) - W_{k}(T(\pi,y),l+1)\big)\, \sigma(\pi,y) \leq 0
		\end{align*}\\
		(b) $\pi \in \stopset_{k}^{l+1}\cap\contset_{k}^{l+2}$. This implies $\pi \in  \stopset_{k}^{l}$ and $\pi \notin \stopset_{k}^{l+2}, \contset_{k}^{l+1},\contset_{k}^{l}$.
		\begin{align*}
		W_{k+1}(\pi,l) - W_{k+1}(\pi,l+1) = \\
		\rho \sum_\obsv \big( W_{k}(T(\pi,y),l+1) \big)\, \sigma(\pi,y) -\reward^\p\compPowerDistribution \leq 0
		\end{align*}
		Last inequality holds using $\compPowerDistribution\in\stopset_k^{l+1}$ and~\eqref{eq:set-definition}\\
		(c) $\pi \in \stopset_{k}^{l}\cap\contset_{k}^{l+1}$. This implies $\pi \in  \contset_k^{l+2}$ and $\pi \notin \stopset_{k}^{l+2}, \stopset_{k}^{l+1},\contset_{k}^{l}$.\begin{align*}
		W_{k+1}(\pi,l) - W_{k+1}(\pi,l+1) = \\
		\rho \sum_\obsv \big( W_{k}(T(\pi,y),l) \big)\, \sigma(\pi,y) -\reward^\p\compPowerDistribution \leq 0
		\end{align*}
		Last inequality holds using $\compPowerDistribution\in\contset_k^{l+1}$ and~\eqref{eq:set-definition}\\
		(d) $\pi \in \contset_{k}^{l}$. This implies $\pi \in  \contset_{k}^{l+1},\contset_{k}^{l+2}$ and $\pi \notin \stopset_{k}^{l+2}, \stopset_{k}^{l+1},\stopset_{k}^{l}$.
		\begin{align*}
		W_{k+1}(\pi,l) - W_{k+1}(\pi,l+1) = \\
		\rho \sum_\obsv \big( W_{k}(T(\pi,y),l) - W_{k}(T(\pi,y),l)\big)\, \sigma(\pi,y) \leq 0
		\end{align*}\\
		Thus we have showed that $W_{k+1}(\pi,l) \leq W_{k+1}(\pi,l+1)$. From the definition of  the continue and stopping sets in~\eqref{eq:set-definition}, it then  follows that $\stopset^{l}_{k+1}\supset \stopset^{l+1}_{k+1}$.

	\end{proof}    
    Next, we define a submodular function. This is required to prove that the mining policy is monotone.
    \begin{definition}[Submodular function]
        A function $f(\pi,u)$ is submodular if $f(\pi, u)-f(\pi, \bar{u}) \geq f(\bar{\pi}, u)-f(\bar{\pi}, \bar{u}) \text { for } u \geq \bar{u}, \pi \geq_r \bar{\pi}$.
    \end{definition}      
    To show the existence of an optimal policy which is decreasing in $\pi\in\scrL(e_i,\bar{\compPowerDistribution}),\;i\in\{1,\numStops\}$ (Theorem~\ref{subthm:monotonicity}), we use the following result on submodular functions from~\cite{2011:DT}:
    \begin{theorem}
        \label{thm:submodular-decreasing}
        If $f(\pi, u)$ is submodular, then there exists an optimal policy $u^*(\pi)=\arg\max_{u \in \mathcal{U}}(\pi, u)$ that is MLR decreasing in $\pi$.  
    \end{theorem}
    We use the MLR order for $\pi$ in the set $\scrL(e_i,\bar{\pi}),\;i\in\{1,\numStops\}$. 
    Hence, all we need to show is that $Q(\pi,\stopNumber,\action)$ is a submodular function. Before deriving the main result, we present a property of the Bayesian update from~\cite{2016:VK}: MLR ordering is preserved under Bayes' rule.
    \begin{theorem}
        If the transition matrix $\transMatrix$ and the observation matrix $\obsvMatrix$ are TP2. Then, for $\pi_1\geq_r\pi_2$, $T(\pi_1,\cdot)\geq T(\pi_2,\cdot)$ and $\sigma(\pi_1,\cdot)\geq \sigma(\pi_2,\cdot)$.
    \end{theorem}
    
    The following result guarantees the existence of a monotonically decreasing mining policy.
    \begin{theorem}
        \label{thm:q-submodular}
        $Q(\pi,\stopNumber,\action)$ is a submodular function $\forall \stopNumber$.
    \end{theorem}
    \begin{proof}  
    Since, $Q_{k+1}(\pi, l, 2)-$ $Q_{k+1}(\pi, l, 1)=-\rho \sum_y W_k(T(\pi, y), l) \sigma(\pi, y)+r^{\intercal} \pi$. We need to show that $-\rho \sum_y W_k(T(\pi, y), l) \sigma(\pi, y)+r^\intercal \pi$ is MLR decreasing in $\pi$.
    $$
    \begin{array}{l}
    -\rho \sum_y W_k(T(\pi, y), l) \sigma(\pi, y)+r^\intercal \pi \\
    =-\sum_y\left(\left(\rho W_k(T(\pi, y), l)+\rho r^\intercal T(\pi, y)\right)\right. \\
    \left.+\left(r^\intercal \pi-\rho r^\intercal T(\pi, y)\right)\right) \sigma(\pi, y) \\
    =-\rho \sum_y\left(W_k(T(\pi, y), l)-r^\intercal T(\pi, y)\right) \sigma(\pi, y)\\ 
    +r^\intercal\left(I-\rho \transMatrix^\intercal\right) \pi
    \end{array}
    $$
    The term $r^\intercal\left(I-\rho \transMatrix^\intercal\right) \pi$ in is MLR decreasing in $\pi$ due to our assumptions in Sec.\ref{sec:assumptions}. Hence, to show that $-\rho \sum_y W_k(T(\pi, y), l) \sigma(\pi, y)+r^\intercal \pi$ is MLR decreasing in $\pi$ it is sufficient to show that $W_k(\pi, l)-r^\intercal \pi$ is MLR decreasing in $\pi$. Define, $\overline{W}_k(\pi, l) := W_k(\pi, l)-r^\intercal \pi$.
    $$
    \begin{aligned}
    V_k(\pi, l) & =\left(r^\intercal \pi+\rho \sum_y V_{k-1}(T(\pi, y), l-1) \sigma(\pi, y)\right) \mathcal{I}_{M_k^{l}}\\ 
    &+\left(\rho \sum_y V_{k-1}(T(\pi, y), l) \sigma(\pi, y)\right) \mathcal{I}_{D_k^{l}}
    \end{aligned}
    $$
    where $\mathcal{I}_{D_k^l}$ and $\mathcal{I}_{M_k^l}$ are indicator functions on the don't mine and mine sets, respectively, for each iteration $k$. As $M_k^{l+1} \subset M_k^l$,
    \begin{align*}
    W_k(\pi, l) & =\left(\rho \sum_y W_{k-1}(T(\pi, y), l) \sigma(\pi, y)\right) \mathcal{I}_{D_k^l}(\pi)\\ 
    &+r^\intercal \pi \mathcal{I}_{D_k^{l+1} \cap M_k^l}(\pi) \\
    & +\left(\rho \sum_y W_{k-1}(T(\pi, y), l+1) \sigma(\pi, y)\right) \mathcal{I}_{M_k^{l+1}(\pi)}\\
    \Rightarrow\overline{W}_k(\pi, l)&=\displaystyle\sum_y\widetilde{W}_{k-1}(T(\pi, y),l) \mathcal{I}_{D_k^{l}}(\pi)\\
    &+\displaystyle\sum_y\widetilde{W}_{k-1}(T(\pi, y),l+1)\mathcal{I}_{M_k^{l+1}}(\pi)\\
    \widetilde{W}_k(\pi,l)&:= \rho \overline{W}_{k}(\pi, l) \sigma(\pi, y)-r^\intercal(I-\rho P)^{\intercal} \pi
    \end{align*}    
    We prove using induction that $\overline{W}_k(\pi, l)$ is MLR increasing in $\pi$, using the recursive relation over $k$. For $k=0, \overline{W}_0(\pi, l)=$ $W_0(\pi, l)-r^\intercal \pi=V_0(\pi, l)-V_0(\pi, l+1)-r^\intercal \pi$. The initial conditions of the value iteration algorithm can be chosen such that $\overline{W}_0(\pi, l)$ is increasing in $\pi$.
    
    Next, we show that $\overline{W}_k(\pi, l)$ is MLR increasing in $\pi$, if $\overline{W}_{k-1}(\pi, l)$ is MLR increasing in $\pi$. For $\pi_1 \geq_r \pi_2$, consider the following cases: (a) $\pi_1, \pi_2 \in M_k^{l+1}$, (b) $\pi_2 \in M_k^{l+1}, \pi_1 \in D_k^{l+1}$, (c) $\pi_1, \pi_2 \in D_k^{l+1} \cap M_k^l$, (d) $\pi_2 \in D_k^{l+1} \cap M_k^l, \pi_1 \in D_k^l$, (e) $\pi_1, \pi_2 \in D_k^l$. For cases (a), (e), $\overline{W}_k\left(\pi_1, l\right) \geq \overline{W}_k\left(\pi_2, l\right)$ by the induction assumption. For case (b), $\overline{W}_k\left(\pi_1, l\right) \geq \overline{W}_k\left(\pi_2, l\right)$ by definition of $M^{l+1}_k$ and $D^{l+1}_k$~\eqref{eq:set-definition}. For case (c), $\overline{W}_k\left(\pi_1, l\right)=\overline{W}_k\left(\pi_2, l\right)=0$.  For case (d),  $\overline{W}_k\left(\pi_1, l\right) \geq \overline{W}_k\left(\pi_2, l\right)$ by definition of $M^{l}_k$ and $D^{l}_k$.
    \end{proof}
    
    Theorem~\ref{subthm:threshold} follows from the monotonicity property. Suppose $M^l$ is a disconnected set. We can find a line $\scrL(e_1,\bar{\compPowerDistribution}),\;i\in\{1,\numStops\}$ that passes through the two disconnected components of $M^l$. Existence of disconnected set would imply that $u^*(\pi)$ is not monotonic in the set $\scrL(e_i,\bar{\compPowerDistribution}),\;i\in\{1,\numStops\}$ with respect to the MLR order. This is a contradiction.
    
    The proof of Theorem~\ref{subthm:nested} follows from Lemma~\ref{lemma:w-increasing}.

    \section{Proof of Theorem~\ref{thm:parameter-condition} in Sec.\ref{sec:simulations}}
    \label{proof:thm-parameter-condition}
    $D^l$ is non-empty implies $e_{\compPowerSetSize-1}$ should be in $D^l$. This implies $\theta_\stopNumber(1)\geq 0, \forall \stopNumber$. We first derive conditions on $\begin{bmatrix}\theta_\stopNumber & 1 & 0\end{bmatrix}\left[\begin{smallmatrix}-1\\\compPowerDistribution\end{smallmatrix}\right]$ to be MLR decreasing on $\scrL(e_i,\bar{\compPowerDistribution}),i\in\{1,\compPowerSetSize\}$. For $\compPowerDistribution_1\geq_r \compPowerDistribution_2$,
        $\begin{bmatrix}\theta_\stopNumber & 1 & 0\end{bmatrix}\left[\begin{smallmatrix}-1\\\compPowerDistribution_1\end{smallmatrix}\right]\leq \begin{bmatrix}\theta_\stopNumber & 1 & 0\end{bmatrix}\left[\begin{smallmatrix}-1\\\compPowerDistribution_2\end{smallmatrix}\right]
        \Leftrightarrow\bracketSquare{\theta(2)\; \ldots\; \theta(\compPowerSetSize-1)\; 1\; 0}(\compPowerDistribution_1-\compPowerDistribution_2)\leq 0$. 
    For $\compPowerDistribution_1,\compPowerDistribution_2\in\scrL(e_\compPowerSetSize,\bar{\compPowerDistribution})$,
        $\bracketSquare{\theta(2)\; \ldots\; \theta(\compPowerSetSize-1)\; 1\; 0}(e_\compPowerSetSize-\bar{\pi})\leq 0
        \Leftrightarrow \bracketSquare{0\; \theta(2)\; \ldots\; \theta(\compPowerSetSize-1)\; 1\; 0}\bar{\compPowerDistribution}\geq 0$.
        For $\compPowerDistribution_1,\compPowerDistribution_2\in\scrL(e_1,\bar{\compPowerDistribution})$,
        $\bracketSquare{\theta(2)\; \ldots\; \theta(\compPowerSetSize-1)\; 1\; 0}(e_1-\bar{\pi})\geq 0
        \Leftrightarrow \theta(2)-\bracketSquare{\theta(2)\; \ldots\; \theta(\compPowerSetSize-1)\; 1\; 0}\bar{\compPowerDistribution}\geq 0$
    
    The last condition on the parameters can be derived by using the nested property of the mining set $M^{\stopNumber}$, $
        \policy_\theta(\compPowerDistribution,\stopNumber)\geq \policy_\theta(\compPowerDistribution,\stopNumber+1)\Leftrightarrow \begin{bmatrix}\theta_\stopNumber-\theta_{\stopNumber+1} & 1 & 0\end{bmatrix}\left[\begin{smallmatrix}-1\\\compPowerDistribution\end{smallmatrix}\right]\geq 0$

\end{document}